\documentclass[a4paper]{amsart}
\usepackage{amsmath}
\usepackage{amssymb}
\usepackage{amsfonts}
\usepackage{amsthm}
\usepackage{amssymb}
\usepackage{amsbsy}
\usepackage{color}
\usepackage{graphicx}
\usepackage{subcaption}
\usepackage[utf8]{inputenc}
\usepackage{csquotes}
\usepackage[greek,english]{babel}

\usepackage{asymptote}

\usepackage{enumitem}
\usepackage[makeroom]{cancel}

\usepackage[bookmarks,pdfencoding=auto]{hyperref}
\newtheorem{theorem}{Theorem}[section]
\newtheorem{corollary}[theorem]{Corollary}
\newtheorem{lemma}[theorem]{Lemma}
\newtheorem{proposition}[theorem]{Proposition}
\newtheorem{assumption}[theorem]{Assumption}

\theoremstyle{definition}
\newtheorem{definition}[theorem]{Definition}
\theoremstyle{remark}
\newtheorem{remark}[theorem]{\textbf{Remark}}
\newtheorem{example}[theorem]{Example}
\numberwithin{equation}{section}

\newcommand{\E}{\mathbb{E}}

\renewcommand{\P}{\mathbb{P}}
\newcommand{\Q}{\mathbb{Q}}

\newcommand{\di}{\mathrm{d}}

\newcommand{\half}{\dfrac{1}{2}}

\newcommand{\R}{\mathbb{R}}

\newcommand{\as}{a.s.}

\newcommand{\ie}{i.e.}
\newcommand{\eg}{e.g.}

\newcommand{\dx}{\mathrm{d}x}

\newcommand{\dt}{\mathrm{d}t}

\newcommand{\F}{\mathcal{F}}
\newcommand{\Fc}{\mathcal{F}}

\newcommand{\eps}{\varepsilon}
\newcommand{\indic}[1]{\boldsymbol{1}_{\{\ensuremath{#1}\}}}
\newcommand{\ind}{\boldsymbol{1}}

\newcommand{\me}{\mathrm{e}}

\newcommand{\Ep}[1]{\E\left[#1\right]}
\newcommand{\Eq}[1]{\E^{\Q}\left[#1\right]}
\newcommand{\Ebq}[1]{\E^{\bQ}\left[#1\right]}
\newcommand{\Edot}[2]{\E^{#1}\left[#2\right]}

\newcommand{\Nb}{\mathbb{N}}

\newcommand{\Pc}{\mathcal{P}}
\newcommand{\Pcc}{\Pc^{\circ}}

\newcommand{\Qc}{\mathcal{Q}}

\newcommand{\Vc}{\mathcal{V}}
\newcommand{\cadlag}{c\`adl\`ag}

\newcommand{\ol}{\overline}
\newcommand{\ul}{\underline}

\usepackage{xcolor}

\definecolor{orange}{rgb}{1,0.3,0.2}

\newcommand{\tC}{\tilde{C}} % Traded call prices
\newcommand{\In}{\mathcal{I}} %Intrinsic value
\newcommand{\Inc}{\In^{\circ}} %Intrinsic value

\newcommand{\OT}{\mathsf{OT}}
\newcommand{\BS}{\mathsf{BS}}

\newcommand{\bQ}{\ol{\mathbb{Q}}}

\title[Utility maximisation with model-independent constraints]{Utility maximisation with model-independent constraints}

\author{Alexander M.~G.~Cox}
\thanks{Alexander M.~G.~Cox, Department of Mathematical
  Sciences, University of Bath, Bath, U.~K..\\ e-mail: \texttt{a.m.g.cox@bath.ac.uk}, web: \texttt{http://www.maths.bath.ac.uk/$\sim$mapamgc/}}
\author{Daniel Hernandez-Hernandez}
\thanks{Daniel Hern\'andez, Department of Probability and Statistics, Research Center for 
Mathematics (CIMAT), Mexico. \\ e-mail: \texttt{dher@cimat.mx}, web: \texttt{http://www.cimat.mx/$\sim$dher/}}
\thanks{}
\keywords{Model-independent constraints, trading restrictions, max-plus decomposition, utility maximization, intrinsic value} 
\date{\today}

\begin{document}

%% For Asymptote:

\begin{asydef}
usepackage("amsmath");
texpreamble("\newcommand{\ul}[1]{\underline{#1}}\newcommand{\ol}[1]{\overline{#1}}\newcommand{\indic}[1]{\boldsymbol{1}_{\{\ensuremath{#1}\}}}\newcommand{\Rt}{\overline{R}}\newcommand{\Rs}{\underline{R}}\newcommand{\BS}{\mathsf{BS}}");

import graph;

real bs(real s, real K, real T, real r, real sigma)
{
  real N(real x)
  {
    return 1/2+erf(x/sqrt(2))/2;
  }

  real d1 = (sigma>0)?(log(s)-log(K) + (r + sigma^2/2)*T)/(sigma*sqrt(T)):10^6*sgn(log(s)-log(K)+r*T);
  real d2 = d1 - sigma*sqrt(T);
  
  return s*N(d1) - N(d2)*K*exp(-r*T);
}
\end{asydef}

\begin{abstract}
We consider an agent who has access to a financial market, including derivative contracts, who looks to maximise her utility. Whilst the agent looks to maximise utility over one probability measure, or class of probability measures, she must also ensure that the mark-to-market value of her portfolio remains above a given threshold. When the mark-to-market value is based on a more pessimistic valuation method, such as model-independent bounds, we recover a novel optimisation problem for the agent where the agents investment problem must satisfy a pathwise constraint.

For complete markets, the expression of the optimal terminal wealth is given, using the max-plus decomposition for supermartingales.  Moreover, for the Black-Scholes-Merton model the explicit form of the process involved in such decomposition is obtained, and we are able to investigate numerically optimal portfolios in the presence of options which are mispriced according to the agent's beliefs.
\end{abstract}
\maketitle

\section{Introduction}\label{sec:Introduction}

In this paper we  consider a utility maximisation problem for an agent who has some modelling beliefs, according to which the agent will aim to maximise her utility, but also some constraints which are based on model-independent considerations. Our basic setting is that the agent assumes they will observe only `possible' paths according to their beliefs, and they will pursue a utility maximisation objective corresponding to their beliefs. We importantly include in our setting both trading in an underlying risky asset, as well as in illiquid derivatives, whose initial price and payoff are known, but no assumptions about the intermediate value can be made. The agent is also being observed by a manager or regulator who does not share the agent's modelling assumptions, but rather uses other (typically more pessimistic) assumptions. The manager will intervene if their valuation of the agent's portfolio goes below some given threshold, and the agent will act to avoid this scenario. Note that many real-world trading strategies are subject to related constraints, for example Interactive Brokers, an electronic trading platform, base customer margin requirements on a \emph{portfolio margin} basis, which they state is `determined using a ``risk-based'' pricing model that calculates the largest potential loss of all positions in a product class or group across a range of underlying prices and volatilities', \cite{InteractiveBrokers:2025}.\footnote{We are grateful to Paulo Guasoni for pointing this out to us.} 

Under these modelling assumptions, our aim will be to determine the agent's optimal trading strategy when they are able to take (static) positions in certain options, for example, call options, or other simple derivatives. In the context of some of these options, we will use the notion of an ``intrinsic'' value of an option, which we think of as the worst-case valuation of the option or portfolio of options, for example (in the absence of interest rates) the intrinsic vale of a long position in a call option with maturity $T$ and strike $K$ at time $t< T$ is $(S_t-K)_+$, since this can be realised through taking model-independent positions in the underlying asset.

Our approach borrows from the literature model-independent or robust pricing and hedging (see e.g.~\cite{Hobson:1998aa,Cox:2011aa,galichon_stochastic_2014,beiglbock_model-independent_2013,dolinsky_super-replication_2018,Dolinsky:2014aa,beiglbock_pathwise_2017,cheridito_martingale_2021,cox_roots_2013,cox_model-independent_2017}). We also use classical results from the theory of utility maximisation in complete markets. Our approach to handling pathwise constraints is heavily inspired by the papers of~\cite{El-Karoui:2006aa,El-Karoui:2008aa}, see also \cite{bank_stochastic_2021,bank_stochastic_2004}.

\subsection{Basic Problem Formulation}

We consider a market on the time interval $[0,T]$, and we suppose that an asset price $(S_t)_{t \in [0,T]} \in C([0,T])$ is observed, where for the moment we suppose all prices are given in discounted units. We suppose the agent believes that there is a class of probability measures $\Pc$, and the agent aims to find
\begin{equation}\label{eq:MainProb}
  \sup_{ \pi} \inf_{ \P \in \Pc} \Edot{\P}{U\left(w_0 + X_T^\pi\right)}, 
\end{equation}
where $U$ is a utility function, $w_0$ the initial wealth of the agent, and $X_T^{\pi}$ the trading gains of the agent given they follow the trading strategy $\pi$. 

More generally, we suppose that at time zero, the trader observes additional market information in the form of the prices of other traded derivatives. For example, the agent may observe the prices of call options given by $\tC(K) = \int (x-K)_+ \, \mu(\dx)$ for some probability measure $\mu$ (given by the Breeden-Litzenberger formula,~\cite{breeden_prices_1978}). In this case, the trader may purchase a portfolio of call options with payoff $h(S_T)$ for price $\int h(x) \, \mu(\dx)$ at time 0.

In our setup, we will impose a `model-independent' restriction on the trader's behavior by assuming a specific budget constraint. This constraint will occur when the trader's portfolio contains derivatives. Our underlying assumption is that, even though the trader may evaluate the `correct' price of the derivative in their model, they are subject to portfolio constraints imposed by a manager or regulator who is much more risk averse, and who values their derivatives using a more conservative set of pricing rules. We will typically call the valuation of the manager or regulator the \emph{intrinsic value} of the derivatives. The canonical example of such a derivative is a call option, where the intrinsic value of the derivative is the `zero-volatility' payoff of the option, or the terminal value of the call option if the asset grows at the (deterministic) interest rate $D_{\cdot}^{-1}$.

Of crucial interest to us is that these intrinsic values are in general non-linear, and so choosing to purchase different portfolios of derivatives will have a complex effect on the terminal wealth of the investor, and hence on the optimal investment strategies of the investor. Our usual setup in this paper will be the case where the intrinsic value of the derivative corresponds to the model-independent sub-replication price of the derivative, and in many examples we are able to specifically identify this quantity in terms of the underlying contract.

For example in the case where the agent may purchase a portfolio of calls with payoff $h(S_T)$ then the \emph{intrinsic value} of the portfolio at time $t$ will be $h^*(S_t)$, where $h^*$ is the greatest convex minorant (on $[0,\infty)$) of the function $h$. We think of the intrinsic value at time $t$, which we write $\In_t(S_T) = h^*(S_t)$, as the minimum value of the portfolio which can be guaranteed under \emph{any} possible model. For example, if $h(x) = (x-K)_+$, then $h^*(x) = h(x) = (x-K)_+$, and we confirm that this amount may be realised at time $t$ through the trading strategy which (if $S_t >K$) short sells the asset until either the asset drops below $K$, or time $T$, whichever is earlier. If we write $H^t_K:= \inf \{r \ge t: S_r \le K\}$, then the value of this portfolio at maturity is:
\begin{align*}
  (S_T-K)_+ + (S_t-S_{T \wedge H^t_K}) 
  & = 
    \begin{cases}
      (S_T - K)_+ + (S_t - K) & H^t_K \le T\\
      (S_T - K) + (S_t-S_T) & H^t_K > T
    \end{cases}
  \\
  & \ge S_t - K.
\end{align*}
By considering the model where the asset price remains constant, so that the price under this model is equal to the intrinsic value, we conclude that this is the best we can do. Note, in particular, that with this trading strategy the trader's wealth (including the intrinsic value) at time $u>t$ is always at least $S_u$: that is,
\begin{equation*}
  \In_u((S_T-K)_+) + (S_t-S_{u \wedge H^t_K}) \ge S_u.
\end{equation*}

The constraint we impose on our trader is that the trader's portfolio must satisfy an admissibility constraint, based on the intrinsic value of the derivatives. Specifically, we require the trader's (intrinsic) wealth at every time $t$ to satisfy:
\begin{equation}\label{eq:Wdefn}
  W_t^{\pi,h} := D_t^{-1}(w_0  - \int h(x) \, \mu(\dx))+ \In_t(h(S_T)) + X_t^{\pi} \ge -\alpha.
\end{equation}
The quantity $\alpha$ represents a lower bound imposed on the trader's portfolio value, which is required to be observed at all times. We will call a wealth process which satisfies \eqref{eq:Wdefn} \emph{$\alpha$-admissible}.

Our intuition is as follows: the trader will follow her trading strategy $\pi$ in a manner that maintains \eqref{eq:Wdefn} at all times. Under any probability measure $\P \in \Pc$, this will result in a portfolio which satisfies the constraint. If the real path of the asset does not follow a path which is compatible with any $\P \in \Pc$, then the trader can monitor her wealth, and at the first time the wealth goes below $-\alpha$, then the trader will stop dynamic trading, and simply follow the simple strategy which realises (at worst) the intrinsic value of the asset. Combined with the intrinsic value of the portfolio, this strategy always ensures that the portfolio's value remains above the lower bound.% \footnote{We might need to be a little careful about how exactly we do this --- we will probably want to include strategies which attain $-\alpha$ at some time before $T$, and can maintain this. In which case we should stop if we ever reach $-\alpha-\eps$ for some small $\eps$, but this is easy to argue, just not so easy to tell a nice story.}

Alternatively, one could tell the story from the perspective of a trader who is being monitored by a manager or regulator. The manager is conservative, and will look at the agent's gains from trade continuously, evaluating their derivative portfolio using a stated, model-independent rule. If the trader's \emph{intrinsic} wealth goes below the level $-\alpha$, the manager will fire the trader and close out the position with a resulting loss bounded below by $\alpha$. As a result, the trader wishes to pursue a strategy which does not result in their dismissal, and hence looks to find a strategy which stays above the intrinsic wealth constraint with probability one (under any model that they believe is possible).

Our results will take two different forms: first we will consider cases where the trader is able to trade dynamically to exploit mispricings, and guarantee a profit under certain conditions on the admissibility level $\alpha$; these results will be in a similar spirit to classical model-independent pricing constraints on traded option prices. Second, we will consider utility maximisation problems, where the trader aims to maximise their utility from terminal wealth under the additional constraint that their wealth process is admissible. In this case, we will examine the impact of different choices of derivative portfolios, and will give concrete conclusions about the optimal strategies that should be employed by the investor when faced with various traded options on the market.

\begin{remark}
  We note that the notion of intrinsic value introduced above is fairly flexible. For example, above we defined the intrinsic value to be the convex minorant on $[0,\infty)$, which corresponds to the case where it is not believed that the asset price can go negative. However, it is possible also to consider the intrinsic value assuming that asset prices can go below zero (e.g. in the Bachelier model). In this case, it might not be sufficient to consider the convex minorant on $[0,\infty)$, but rather to look at the minorant on $\R$. In some cases, this would give different values of the intrinsic process. More generally, intrinsic values arising from robust pricing bounds (i.e. over a class of models) as opposed to model-independent (over all models) bounds are also natural to consider.
\end{remark}

This paper is organized as follows: Section 2 introduces the notion of intrinsic value of derivative contracts  in terms of a subreplicating submartingale, providing examples for specific
cases.   The robust optimization problem of maximizing  expected utility of terminal wealth subject to model independent  intrinsic budget restrictions is presented in Section 3, together with some implications in the trader's behavior. Specific results are obtained for the Black-Scholes-Merton model. In Section 4 optimal trading strategies are obtained under the assumption of completeness of the market, with the help  of representations of supermartingales.

\section{Intrinsic Valuation of Derivatives and Trading}

We define here our basic market setup. We suppose that there is an underlying asset price process $(S_{t})_{t \in [0,T]}$ which takes values in $\Omega: = C_{s_0}([0,T])$, the set of continuous paths $\omega$ on $[0,T]$ with $\omega(0) = s_0$, and where we equip $\Omega$ with the uniform norm, under which topology $\Omega$ is a Polish space. This space is endowed with the Borel $\sigma-$algebra $\mathbb{F}$. Our agent believes that the underlying dynamics of $S$ are governed by a probability measure $\P \in \mathcal{P}$, for some class $\mathcal{P}$ of probability measures on $\Omega$. We will typically be interested in statements which hold $\P$-\as{} for all $\P \in \Pc$, which we will write as $\Pc$-\as{}. It is natural therefore to introduce $\mathcal{N}(\mathcal{P}):= \{A \in \mathbb{F}: \P(A) = 0 \ \forall \P \in \Pc\}$. We also introduce the set $\Qc(\Pc)$, the set of martingale measures which are equivalent to some $\P \in \Pc$. The natural filtration generated by $S$ is denoted by $\Fc^S=\{\F_t^S\}$. We also need the filtration $\Fc=\{\F_t;\;t\in[0,T]\}$, with $\F_t:=\F^S_{t^+}=\cap_{s>t} \F_s^S$ for $t<T$ and $\F_T:=\F_T^S$, which is the minimal filtration associated to the process $S$ satisfying the usual conditions, i.e. $\{\F_t;\;t\in[0,T]\}$ is an increasing right continuous family of $\sigma-$fields and completeness with respect to $\mathcal{P}$, by which we mean that $\mathcal{N}(\Pc) \subset \F_0$.
%and $\F_0$ contains the $\Pc-$negligible sets in $\mathbb{F}$. 
We denote by $\Lambda$ a  Meyer $\sigma-$field which contains the predictable $\sigma-$field with respect to $\F$, that is, the  $\sigma-$field generated by $\F-$adapted, left-continuous processes, and which in turn is contained in the optional $\sigma-$field with respect to the filtration $\F$. 

In addition to the risky asset, we suppose there exists a bank account which pays a deterministic (although not necessarily constant) interest rate. We write $D_t$ for the discount factor, so the time-0 value of $\$1$ at time $t$ is $D_t$, or equivalently, $\$1$ invested at time $0$ will be worth $\$D_t^{-1}$ at time $t$. We assume then that $D_t$ is decreasing, continuous  and $D_0 = 1$.

We also associate with our setup   trading strategies $\pi$ with respect to the filtration $\F$. In this paper we do not wish to directly address the specific technicalities of possible trading strategies under model-uncertainty, but refer readers to the large and growing literature for various approaches (\eg{} \cite{Dolinsky:2014aa, beiglbock_pathwise_2017, biagini_robust_2017, hou_robust_2018, dolinsky_martingale_2015,cheridito_martingale_2021} among others). The details here will not be important, so we will generally either work in the case where $\Pc$ is a singleton, and classical results are applicable, or else in the case where $\Pc$ is large, and then we will only need to consider very simple trading strategies; see Remark \ref{Rem:Est2.6} below.

Note in particular, that these two cases are essentially the main ones of interest. For example, \cite{dolinsky_super-replication_2018} show that a robust hedging for a large class of stochastic volatility models essentially reduces to the case where $\Pc$ contains all martingale measures.

We also consider the special `classical' case of model-independent pricing, where
\begin{equation}\label{def:P0}
  \Pcc := \{ \P: D S \text{ is a non-negative, uniformly integrable martingale}\}.
\end{equation}
This will give rise to our canonical notion of intrinsic value, but other choices will also be possible.

\subsection{Intrinsic Value of Derivative Contracts}

We consider the {\it intrinsic} valuation of a derivative contract:

\begin{definition} \label{def:intrinsic}
  A \emph{derivative contract} is a measurable function $C_T: \Omega \to \R$. We say that $\In_{\cdot}(C_T)$ is a \emph{fair intrinsic value} of a derivative contract $C_T$ corresponding to the class $\Pc$ of probability measures at time $t$, if:
  \begin{enumerate}
%  \item $D_t\In_t(C_T) \le \Eq{D_T C_T|\Fc_t}$, for all $\Q \in \Qc(\Pc)$; \comment{redundant}
  \item $D_t\In_t(C_T)$ is a \cadlag{} $\Q$-$\F$-submartingale for all $\Q \in \Qc(\Pc)$;
  \item $\In_T(C_T) = C_T$ $\Pc$-\as{}.
  \end{enumerate}
\end{definition}

As an interesting fact, of course, $D_t\In_t(C_T) \le \Eq{D_T C_T|\Fc_t}$, for all $\Q \in \Qc(\Pc)$, where we are assuming  throughout that $D_T C_T$ satisfy implicitly the required integrability conditions in order that the conditional expectation is well defined. On the other hand, in general, we would hope to find a maximal version of the fair intrinsic price for a given set $\mathcal{P}$, which one could define (except for non-trivial measurability issues!) to be the price of the most expensive model-independent sub-replicating strategy. That is:
\begin{equation}\label{eq:6}
 D_t \In_t(C_T) := \sup \{ x \in \R: \exists (\pi) \text{ s.t. } x + \int_t^T \pi_r \, \di (D_r S_r ) \le D_T C_T \quad \Pc-\as{}\}.
\end{equation}
Such problems have been considered recently in discrete time (\cite{beiglbock_model-independent_2013,bouchard_consistent_2016,aksamit_robust_2020}) and continuous time (\cite{Dolinsky:2014aa, beiglbock_pathwise_2017, biagini_robust_2017, hou_robust_2018, dolinsky_martingale_2015,cheridito_martingale_2021}), but defining this process in general in continuous time is a non-trivial technical exercise. For the majority of this paper, our aim will be to consider easily specified intrinsic value processes, but we emphasise that in our setup the chosen fair intrinsic value is a part of the modelling framework, and not necessarily a given quantity.

Note in particular that we do not expect the intrinsic price $\In_t$ to be linear: we do not typically expect for example $\In_t(\beta C_T) = \beta \In_t(C_T)$ if $\beta<0$. However, the intrinsic price will generally be positive homogenous, with $\In_t(\beta C_T) = \beta \In_t(C_T)$ if $\beta \ge 0$.

\begin{example} \label{ex:Intrinsic}
  In the case where $\Pc = \Pcc$, see (\ref{def:P0}), and $D_t \equiv 1$, we can give concrete examples of an intrinsic value process which does in fact satisfy \eqref{eq:6}. Since many of our examples will be based on this specific choice of the intrinsic price process, we denote this specific operator by $\Inc$.
  \begin{enumerate}
  \item\label{item:1} If $h:\R_+ \to \R$, then $\Inc_t(h(S_T)) = h^*(S_t)\ind_{t<T}
    + h(S_T)\ind_{t=T}$, where $h^*(x)$ is the greatest convex minorant of $h$ (on $\R_+$). For example, $h^*((x-k)_+) = (x-k)_+$, but $h^*(-(x-k)_+) = -x$.

    To see this, we first observe that this intrinsic price satisfies all of the conditions of Definition~\ref{def:intrinsic}. Second, to see that this is the greatest such price, observe that if $h^*(x)$ is the greatest convex minorant of $h$, for $x \in (0,\infty)$ we can find $y_n \le x \le z_n$ such that \[h^*(x) = \lim_{n \to \infty} \left[ \frac{(z_n-x) h(y_n) + (x-y_n)h(z_n)}{z_n-y_n}\right].\]
Now consider the model which is a (uniformly integrable, continuous) martingale, which runs from $x$ at time $t<T$, to either $z_n$ or $y_n$ at time $T$. Then the fair price of the derivative under this model is exactly $ \left[ \frac{(z_n-x) h(y_n) + (x-y_n)h(z_n)}{z_n-y_n}\right]$, and the claim follows. In the presence of non-zero interest rates, it is easy to see by a similar argument that $\Inc_t(h(S_T)) = D_t^{-1} D_Th^*( D_T^{-1} D_tS_t)\ind_{t<T}
    + h(S_T)\ind_{t=T}$.
  \item\label{item:2} If $0 < T' < T$ then $\In_t(h(S_{T'})) = h^*(S_t) \ind_{t < T'} + h(S_{T'}) \ind_{t \ge T'}$. This is essentially the same argument as  in  \emph{(i)}.
  \item\label{item:3} If $B >0$ is a fixed barrier, we can consider the
    \emph{one-touch} option, $\OT^B_T := \ind_{S_T^* \ge B}$, where
    $S_t^* = \sup_{r \le t} S_r$ is the maximum process. In
    particular, it can be checked that one has $\In_t(\OT_T^B) =
    \ind_{S_t^* \ge B}$, while $\In_t(-\OT_T^B) = -\frac{S_t}{B}
    \ind_{S_t^* < B, t<T} - \ind_{S_t^* \ge B}$.

    In fact, in the last case, we can extend some of these ideas. For example, let $K<B$. Then we have for $t \le T$
    \begin{equation*}
      \In_t\left(\frac{(S_T-K)_+}{B-K} - \OT_T^B\right) =
      \begin{cases}
        0, &\quad \text{ if } S_t^* < B\\
        \frac{1}{B-K} \left[ (S_t - S_{H_B}) + (K-S_t)_+\right], & \quad \text{ if } S_t^* \ge B
      \end{cases},
    \end{equation*}
    where we write $H_B := \inf \{ t \ge 0: S_t = B\}$.  This is a consequence of the hedge for one-touch options given by \cite{Hobson:1998aa}.

    We note that the amount $(S_t - S_{H_B})$ is easily constructed from an adapted trading strategy (buy one unit of the asset if it hits $B$ before $T$). Including this additional trading, the intrinsic value of the combined position is then simply $(K-S_t)_+ \ind_{H_B \le t}$.
  \end{enumerate}
\end{example}

%\comment{More details needed here...? More examples. Proofs?}

\subsection{Dynamic Trading Strategies}

We also wish to consider the class of dynamic trading strategies which are available to our agent for investment. Typically we would expect these to be specified as part of the modelling assumptions, and could depend on the choice of $\Pc$. For example, if $\Pc$ is a singleton, one may be able to use the standard stochastic integral, while for more complex choices of $\Pc$, one needs to be more careful to admit a measurable choice of the resulting trading strategy.

\begin{definition}
  To each choice of \emph{dynamic trading strategy} $\pi$, we associate a corresponding \emph{gains process} $X_\cdot^\pi: \Omega \to C_0([0,T])$. We say that a dynamic trading strategy $\pi$, or equivalently the gains process $X^\pi_\cdot$, is \emph{$\Pc$-admissible}, if there exist $a \ge 0$ and  a process $\Gamma$, with $\Gamma_t \ge 0$ $\Pc_t$-\as{} for all $t \ge 0$, such that $\sup_{0 \le t \le T} \Gamma_t$ is $\Q$-integrable for any $\Q \in \Qc(\Pc)$, and  $D_t X_t^\pi \ge -a(1+\Gamma_t)$ $\Pc$-\as{} holds for all $t$.
\end{definition}
By abuse of notation, we will often write $\int_0^t \pi_s \, \di(D_sS_s) = D_t X_t^\pi$, despite the fact that $\pi$ may not be explicitly defined as a pathwise object.
%\begin{assumption}
 % There exists a process $\Gamma, \Gamma_t \ge 0$ $\Pc_t$-\as{} for all $t \ge 0$ such that $\sup_{0 \le t \le T} |\Gamma_t|$ is $\Q$-integrable for any $\Q \in \Qc(\Pc)$.
%\end{assumption}
In most examples, process $\Gamma$ will be taken without further comment as $\Gamma_t = 1 + D_t |S_t|$, and then Doob's inequality will give the required integrability provided 
$$\Eq{((1+D_T S_T) \log (1+D_T|S_T|))_+} < \infty,\;\;\;\text{ for all}\;\;\; \Q \in \Qc(\Pc).$$
Note that we need a constraint on dynamic trading strategies to rule out possible doubling strategies. It is natural to impose conditions which are pathwise and not probabilistic since we are potentially considering multiple pricing measures. Relevant results which show that the pathwise interpretation is sufficient can be found in \cite{acciaio_trajectorial_2013}. The condition above will naturally be satisfied if (for example) our market includes derivatives whose payoffs have growth rate which is larger than $((1+D_T S_T) \log (1+D_T|S_T|))_+$. We can then deduce from the relevant pathwise condition that the trading portfolio is uniformly integrable. If we have stronger integrability (for example we know $\Eq{S_T^p} < \infty$ for some $p>1$), then we can weaken the pathwise constraint by increasing $\Gamma$ appropriately.

We make the following definition:
\begin{definition} \label{def:dynam-trad-strat}
  We say that $\Vc$ is a set of \emph{admissible dynamic trading strategies} if, for each $\pi \in \Vc$, the gains process $X^\pi$ is:
  \begin{enumerate}
  \item $\F$-adapted,
  \item $\Pc$-admissible, and
  \item $D_t X_t^\pi$ is a $\Q$-local martingale for every $\Q \in \Qc(\Pc)$.
  \end{enumerate}
\end{definition}

We further assume that the set $\Vc$ is closed under addition and non-negative scalar multiplication, that is, if $\pi, \psi \in \Vc$ and $\lambda, \mu \ge 0$, then $\lambda \pi+\mu \psi \in \Vc$, where $X^{\lambda\pi+\mu\psi}:= \lambda X^\pi + \mu X^\psi$.

\begin{remark}\label{Rem:Est2.6}
  Note that in the case where the asset price process $S$ is non-negative, the class of simple trading strategies, ie $\pi_t = \sum_{i = 1}^n \, \indic{t \in (\tau_{i-1},\tau_i]} a_i$, where $a_i$ are bounded and  $\Fc_{\tau_i}$-measurable, and $(\tau_i)_{i=1}^n$ is a sequence of increasing, $\F$-predictable stopping times in $[0,T]$, belongs to the set of admissible strategies. Similarly, the \emph{buy-and-hold} strategy is always allowed in a set of admissible dynamic trading strategies. 
\end{remark}

We can adopt results from Dolinsky-Soner, e.g.~\cite{Dolinsky:2014aa} to show that the class of progressively measurable trading strategies includes  
finite variation strategies $\pi$, defining the stochastic integral pathwise via integration by parts: 
$$
\int_0^t \pi_s \, \di(D_sS_s) :=\pi_t D_tS_t-\pi_0 S_0-\int_0^tD_sS_s d\pi_s.
$$
This definition is consistent with the stochastic integration for simple trading strategies.

\section{Trading with constraints on Intrinsic Wealth}

\subsection{General setup and preliminary results} \label{sec:setup}

We begin by making the following observation about trading under intrinsic value constraints. We fix a class $\Pc$ of possible probability measures, and a set $\Vc$ of admissible dynamic trading strategies. We first suppose that there is an agent who acts to maximise utility from wealth, and the utility function $U$ satisfies the Inada conditions ($U:[0,\infty) \to \R\cup \{-\infty\}$ is concave increasing and $U'(0) = \infty, U'(\infty) = 0$). We suppose that the agent will maximise worst case utility, in the presence of a derivative $C_T$ which has been purchased for price $c_0$, so the problem is to find:
\begin{equation}\label{eq:MainProb2}
  \sup_{ \pi\in\Vc} \inf_{ \P \in \Pc} \Edot{\P}{U\left(D_T^{-1}(w_0 - c_0) + C_T + X_T^\pi\right)}, 
\end{equation}
where $X_t^\pi$ is the gains from dynamic trading, subject to the \emph{intrinsic budget constraint}
\begin{equation}\label{eq:Wdefn2}
  W_t^{\pi,C} := D_t^{-1}(w_0-c_0) + \In_t(C_T)  + X_t^{\pi} \ge -D_t^{-1} \alpha, \qquad \Pc-\as{}, \forall\; 0\le t \le T,
\end{equation}
where $w_0, \alpha \ge 0$.

\begin{definition}\label{def:admissset}
  We call a trading strategy $\pi\in \Vc$ a $(w_0, \alpha, c_0,C_T)$-intrinsically admissible dynamic trading strategy if $W_T^{\pi,C}\geq 0$ and \eqref{eq:Wdefn2} holds. Then, we write $\pi \in \Vc(w_0,\alpha,c_0,C_T)$.
\end{definition}

Our first result describes some simple cases where the trader's behavior can be easily described.

\begin{lemma}\label{lem:BasicHedging}
  \begin{enumerate}
  \item\label{item:4} Suppose there exists $\psi \in \Vc$ which $\Pc$-superreplicates $C_T$ for initial value $\kappa$:
    \begin{equation*}
      \kappa + \int_0^T \psi_r \, \di (D_r S_r) \ge D_T C_T, \quad \Pc-\as{},
    \end{equation*}
    and further
    \begin{equation}\label{eq:1}
      c_0 + \int_0^t \psi_r \, \di (D_r S_r) \ge D_t \In_t(C_T), \quad \forall t, \Pc-\as{}
    \end{equation}
    Then if $\kappa < c_0$, the market price of the option, it is never optimal for the trader to purchase the option, that is, for all $\pi \in \Vc(w_0,\alpha,c_0,C_T)$ there exists $\hat{\pi} \in \Vc(w_0,\alpha, 0,0)$ such that $W_T^{\hat{\pi}} \ge W_T^{\pi,C}$.
  \item\label{item:5} Suppose there exists a strategy $\psi \in \Vc$ which $\Pc$-subreplicates $C_T$ for initial value $\kappa$:
    \begin{equation*}
      \kappa + \int_0^T \psi_r \, \di (D_r S_r) \le D_T C_T, \quad \Pc-\as{},
    \end{equation*}
    such that $-\psi\in \Vc$ and  the path constraint:
    \begin{equation*}
      c_0 +  \int_0^t \psi_r \, \di (D_r S_r) \le D_t \In_t(C_T),  \quad \forall t, \Pc-\as{}
    \end{equation*}
    holds.  If $\kappa > c_0$, then the trader can find portfolios with arbitrarily large utility.
  \end{enumerate}
\end{lemma}

\begin{proof}
  \begin{enumerate}
  \item We compare the strategy which purchases the option for price $C_0$, and follows the $(w_0, \alpha, c_0,C_T)$-admissible dynamic trading strategy $\pi \in \Vc(w_0,\alpha,c_0,C_T)$, with the strategy which simply follows the dynamic trading strategy $\pi + \psi$. Since $\pi$ and $\psi$ are both in $\Vc$, so too is the combined trading strategy. We need to show that $\pi + \psi \in \Vc(w_0,\alpha,0,0)$ and $W_T^{\pi+\psi,0} \ge W_T^{\pi,C}$.

    Now $W_T^{\pi,C} = D_T^{-1}(w_0-c_0) + C_T + X_T^{\pi} \ge 0$ since $\pi \in \Vc(w_0,\alpha,c_0,C_T)$, and $W_T^{\pi+\psi,0} = D_T^{-1} w_0 + X_T^{\pi+\psi} = D_T^{-1} w_0 + X_T^{\pi} + X_T^{\psi} \ge D_T^{-1} (w_0- \kappa) + C_T + X_T^{\pi}  > W_T^{\pi,C}$. It remains to show that the strategy $\pi + \psi$ satisfies \eqref{eq:Wdefn2}. Since $\pi \in \Vc(w_0,\alpha,c_0,C_T)$, we know that
    \begin{align*}
      -D_t^{-1} \alpha
      & \le D_t^{-1}(w_0-c_0) + \In_t(C_T)  + X_t^{\pi} \\
      & \le  D_t^{-1}(w_0-c_0) + \left(X_t^\psi + D_t^{-1} c_0\right)  + X_t^{\pi}\\
       & = D_t^{-1}w_0 + X_t^{\pi+\psi} 
    \end{align*}
    for all $t, \Pc-\as{}$, using \eqref{eq:1} in the second line.

  \item Let $\lambda >0$ and suppose the trader buys $\lambda$ units of the derivative for price $c_0$, and takes short position $-\lambda \psi_t$. Then the traders terminal wealth is:
    \begin{equation*}
      W_T^{-\lambda \psi, \lambda C} = D_T^{-1} w_0 + \lambda\left( C_T - X_T^\psi - D_T^{-1} \kappa\right) + \lambda (\kappa - c_0).
    \end{equation*}
    Since first term in brackets is non-negative, and the second term is strictly positive, the trader's utility can be made arbitrarily large as $\lambda \to \infty$.

    On the other hand, at any time $t$, we have:
    \begin{equation*}
      D_t^{-1}w_0 + \lambda \In_t(C_T) - \lambda X_t^{\psi} - \lambda D_t^{-1} c_0 \ge D_t^{-1}w_0\ge -D_t^{-1} \alpha,
    \end{equation*}
    and so the strategy satisfies the trading constraint.
  \end{enumerate}
\end{proof}

\begin{remark}
  \begin{enumerate}
  \item[(a)] Note that Lemma~\ref{lem:BasicHedging} can also be applied to shorting an option, by replacing $C$ by $-C$, and $c_0$ by $-c_0$.
  \item[(b)] In (\ref{item:4}), the pathwise condition is fairly weak: we know $\kappa + D_T X_T^\psi \ge D_T C_T$, and hence for any $\Q \in \Qc(\Pc)$, using Definitions~\ref{def:dynam-trad-strat} and \ref{def:intrinsic}, we expect that
    \begin{align*}
      \kappa + D_t X_t^{\psi} \ge \kappa + \Eq{D_T X_T^{\psi}|\Fc_t} \ge  \Eq{D_T C_T|\Fc_t} \ge \In_t(C_T) .
    \end{align*}
    This is almost sufficient to deduce the pathwise constraint, however there is no guarantee in our setup that this may hold $\Pc$-almost everywhere for a given $t$; there may exist sets of paths which are $\Pc$-possible, but do not appear under any $\Q \in \Qc(\Pc)$ for superhedging problems. See \eg{} \cite{aksamit_robust_2020}  for a discussion of this phenomena.

    On the other hand, the corresponding condition in (\ref{item:5}) is much stronger, in particular taking $t=0$ it already implies that the market price of the option is below its intrinsic value. This reflects the much stronger conclusion possible in (\ref{item:5}).
    
  \end{enumerate}
      
\end{remark}

% The result has the following particularly nice interpretation in the case where $\Pc$ is a singleton, and $\P \in \Pc$ is complete, so $C_T$ can be hedged exactly, for a price $\kappa$ say: if the market price $c_0 \ge \kappa$, then the option is overpriced, and the trader is better off holding the replicating portfolio. On the other hand, if the price $c_0$ is lower than the replication cost, the trader may find benefit in purchasing the option. Further, if the time-0 price of the option is sufficiently low, the option is very good deal for the trader, and in fact they can account for the initial profit in their portfolio value to ensure they can scale the hedge without violating the trading constraint. Note that for this to be the case, it is necessary that $c_0 \le \In_0(C_T)$, or equivalently that the market price of the option is below the intrinsic value of the option --- in the case where the intrinsic value is given by the price of the most expensive, model-independent sub-replicating portfolio, then this implies that this portfolio is the required strategy. In particular, in the case where the market price of the option is below the intrinsic value, the agent can realise an unbounded utility from trading in the option.

\subsection{European Options in the Black-Scholes-Merton model}

For motivation, we start by considering the case where the agent believes the underlying model is the Black-Scholes-Merton model. In this case, the set $\Pc$ is a singleton, and moreover, the market is complete, so a desired (non-negative, integrable) terminal wealth $X_T$ can be realised through an admissible dynamic trading strategy, with portfolio process $\pi \in \Vc$ such that $X_t^\pi = D_t^{-1} \Eq{D_T X_T|\Fc_t}$, where $\Q$ is the usual (unique) risk-neutral measure. We will assume that the intrinsic price value is given by $\In = \In^\circ$ as described in Example~\ref{ex:Intrinsic}.

We suppose that our agent holds a European option with terminal value $h(S_T)$, for some measurable function $h$ such that $h(x) \le a(1+x)$ for some $a>0$, and we are interested in the admissibility of this simple trading strategy for different values of $\alpha$. We can write $\BS(h,t,T,S_t)$ for the price at time $t$ of an option with payoff $h$, time-to-maturity $T$ and current asset price $S_t$. If the trader purchases the option and trades dynamically to hedge the risk completely, it follows that $X_t^\pi = D_t^{-1}\BS(h,0,T,S_0) - \BS(h,t,T,S_t)$. If the trader follows the strategy of investing in the portfolio $h$ at price $c_0$ and hedging using the strategy $\pi$, then the intrinsic portfolio value at time $t$ is given by
\begin{align*}
  W_t^{\pi,h} & := D_t^{-1} (w_0-c_0) +  \In_t(h(S_T))  + X_t^{\pi} \\
              & \;= D_t^{-1} (w_0-c_0)  + \frac{D_T}{D_t} h^*\left(S_t \frac{D_t}{D_T}\right) + D_t^{-1}\BS(h,0,T,S_0) - \BS(h,t,T,S_t).
\end{align*}

Consider initially the case where $h(x) = (x-K)_+$, for $K \ge 0$, write $\BS(h,\cdot,\cdot,\cdot) = \BS^C(K,\cdot,\cdot,\cdot)$ and denote by $c_0(K)$ the time-0 market price of the option, which (since we are taking a long position) we expect (but do not need) to be lower than the fair price of the derivative, \ie{} $c_0(K) < \BS^C(K,0,T,S_0)$. Write $\Delta C(K) := \BS^C(K,0,T,S_0) - c_0(K)$, the difference between the fair price and the market price, then our admissibility criteria (\ref{eq:Wdefn2}) for this strategy becomes:
\begin{align*}
   D_t^{-1}\left(w_0 + \Delta C(K)\right)  + \left(S_t-K \frac{D_T}{D_t}\right)_+ - \BS^C(K,t,T,S_t) & \ge -\frac{\alpha}{D_t}\\
   \iff  \quad  \left(w_0 + \Delta C(K) \right) + \left(S_tD_t-K D_T\right)_+ - D_t\BS^C(K,t,T,S_t) & \ge -\alpha.
\end{align*}
Now, noting that $t \mapsto D_t \BS^C(K,t,T,sD_t^{-1})$ is decreasing by Jensen's inequality and the convexity of $s \mapsto \BS^C(K,t,T,s)$, we have:
\begin{align*}
  \inf_{\substack{t \in [0,T]\\ s \in [0,\infty)}}
  \Big\{(sD_t- KD_T&)_+  - D_t\BS^C(K,t,T,s)\Big\} \\
  & = \inf_{\substack{t \in [0,T]\\ s \in [0,\infty)}}
  \left\{(s-KD_T)_+ - D_t\BS^C(K,t,T,sD_t^{-1})\right\} \\
  & = \inf_{ s \in [0,\infty)}
  \left\{(s-KD_T)_+ - \BS^C(K,0,T,s)\right\} \\
& = -\BS^C(K,0,T,KD_T).
\end{align*}

We can summarise in the following:
\begin{proposition}
  In the Black-Scholes-Merton problem with intrinsic price given by $\In^\circ$, for $\pi$ the usual delta-hedging of a long position in the European call option with strike $K$, then $\pi \in \Vc(w_0, \alpha, c_0(K), (S_T-K)_+)$ if
  \begin{equation} \label{eq:9}
    w_0 + \Delta C(K)  \ge \BS^C(K,0,T,KD_T) -\alpha,
  \end{equation}
  where $\Delta C(K) := \BS^C(K,0,T,S_0) - c_0(K)$.
  Similarly $$-\pi \in \Vc(w_0, \alpha, -c_0(K), -(S_T-K)_+)$$ if
  \begin{equation*}
    w_0 + \alpha - \Delta C(K) \ge K D_T.
  \end{equation*}
\end{proposition}

\begin{proof}
  The first part of the result was shown in the discussion preceding the proposition. To prove the second part of the claim, we sell the call and hedge dynamically. Then we have $h(x) = -(x-K)_+$, $h^*(x) = -x$, and so $\In_t(-(S_T-K)_+) = -S_t$. Then
  \begin{align*}
    W_t^{-\pi,h} & := D_t^{-1} (w_0+c_0(K)) + \In_t(-(S_T-K)_+) + X_t^{-\pi} \\
                & \; = -S_t + \BS^C(K,t,T,S_t) + D_t^{-1} (w_0 + c_0(K) - \BS^C(K,0,T,S_0)).
  \end{align*}
  As above, we need to compute:
  \begin{align*}
    \inf_{\substack{t \in [0,T]\\ s \in [0,\infty)}} \Big\{-sD_t +  D_t\BS^C(K,t,T,s) \Big\}
    &  = \inf_{t \in [0,T]} \left\{D_t\inf_{s \in [0,\infty)}\left[ -s + \BS^C(K,t,T,s)\right] \right\}\\
    & = \inf_{t \in [0,T]} \left\{ - D_t \left(K D_T D_t^{-1}\right)\right\}\\
    & = - K D_T.
  \end{align*}
  The conclusion follows.
\end{proof}

% and therefore the strategy $\pi \in \Vc(w_0, \alpha, c_0(K), (S_T-K)_+)$ if
% \begin{equation*}
%   w_0 + \Delta C(K)  \ge \BS^C(K,0,T,KD_T) -\alpha.
% \end{equation*}

The term $\BS^C(K,0,T,KD_T)$ on the right-hand side of \eqref{eq:9} is the price of an at-the-money call option with strike $K$, and by the Black-Scholes formula, we can rewrite this as $K D_T \left( 1 - 2 \Phi(\sigma \sqrt{T}/2)\right)$. Suppose that  $w_0 + \alpha>0$ and the call options are underpriced. Then we can always find some small $K$, $K< K_+(w_0 + \alpha) := (w_0 + \alpha) D_T^{-1} \left( 1 - 2 \Phi(\sigma \sqrt{T}/2)\right)^{-1}$ such that we can buy the call, and hedge dynamically to guarantee a profit. For larger strikes it will not be possible to follow this strategy unless the mis-pricing is sufficiently large. A similar behaviour is observed when the prices are too large, but now the asymmetry in the intrinsic value of the call options makes the critical strike $K< K_-(w_0 + \alpha) := (w_0 + \alpha) D_T^{-1}$.

\subsection{Consistency of market prices under constrained trading}

Above, we considered only the case where single call options were traded. In reality call options at a range of strikes and maturities are available for trading, and one natural question is whether the prices are consistent. Simple model-free conditions for the absence of arbitrage are well understood, based on simple model-independent arbitrage strategies which can enforce such an arbitrage. In this section, we analyse whether these strategies are available to a trader whose strategies are subject to the admissibility criteria proposed above.

A common setup is to consider the case where call options with all strikes at a given maturity are traded. Then the prices are free of (model-independent) arbitrage only if the market prices for call options, $C(K)$, satisfies the conditions: \emph{(i)} $C$ is convex; \emph{(ii)} $C$ is decreasing; \emph{(iii)} $C(0) = S_0$; \emph{(iv)} $C'_+(0) \ge -D_T$; moreover, it is commonly assumed that also \emph{(v)} $C(K) \to 0$ as $K \to \infty$. The first two conditions can classically be enforced by simple arbitrage. In this section, we show that there exist trading strategies in $\Vc$, which satistfy \eqref{eq:Wdefn2} for $\alpha = 0$, and generates a strictly positive wealth if any of \emph{(i)}--\emph{(iv)} fail. We note that \emph{(v)} is generally more subtle; see \eg{} \cite{Cox:2016aa}, but under this further assumption (\eg{} \cite{Cox:2011aa}), it follows that there exists a probability measure $\mu$ on $\R_+$ such that $C(K) = \int (x-K)_+ \, \mu(\dx)$. We consider a weaker version, which is simply to enforce positivity, which normally follows from the limiting behaviour and the decreasing property.

\begin{lemma}
  Suppose that $\Pc$ is given by $\Pc^\circ$, $\In$ is given by $\In^\circ$, and European call options with strike $K$ and maturity $T$ can be traded at price $C(K)$ at time $0$. Suppose that any of \emph{(i)} $C$ is convex; \emph{(ii)} $C$ is decreasing; \emph{(iii)} $C(0) = S_0$; \emph{(iv)} $C'_+(0) \ge -D_T$; \emph{(v)} $C$ is non-negative; fail.  Then there exists a portfolio of call options with payoff $g(S_T) = \sum_{i=1}^k a_i (S_T-K_i)_+$ and price $g_0 = \sum_{i=1}^k a_i C(K_i)$ with $k \in \Nb$, $a_i \in \R$, and $\eps>0$ such that $\Vc(-\eps,0,g_0,g)$ is non-empty.
\end{lemma}

Note that the conclusion of the lemma, $\Vc(-\eps,0,g_0,g)$ is non-empty, is equivalently a formulation of arbitrage in our setting: that is, there exists a portfolio and trading strategy which can be setup with initial capital $-\eps$, and which will never use our `intrinsic capitalisation' capacity, $\alpha$, but will finish with a non-negative wealth. It is easy to check that the strategies we implement are in fact scalable, so that we can in fact find such a strategy for an arbitrary $\eps >0$ by a simple scaling argument.

\begin{proof}
  Suppose that $C(K)$ is not convex, then there exists $K_1 < K_2 < K_3$ such that $C(K_2) > \lambda C(K_1) + (1-\lambda) C(K_3)$, where $\lambda = (K_3-K_2)/(K_3-K_1)$. Then the agent should define the function $g$ such that $a_1 = \lambda$, $a_3 = (1-\lambda)$, and $a_2 = -1$.

  Choose $\eps:= -g_0 = C(K_2) - \lambda C(K_1) - (1-\lambda) C(K_3)$. The agent holds a portfolio of calls with positive payoff
  \begin{equation*}
    g(S_T) := 
    \begin{cases}
      0 & S_T \not\in( K_1,K_3)\\
      \lambda (S_T-K_1) & S_T \in (K_1,K_2]\\
      (1-\lambda)(K_3-S_T) & S_T \in (K_2,K_3)
    \end{cases}
  \end{equation*}
  So we have $g^*(S_t) = \In_t(g(S_T)) \equiv 0$ for $t<T$. Taking the dynamic trading strategy $\pi$ which is identically zero, we see that our portfolio intrinsic value is $W_t^{\pi,g} = g(S_T)\ind_{t=T}$, which is non-negative under the assumption that $C$ is not convex, and hence $\pi = \in \Vc(-\eps,0,g_0,g)$.

  The cases (ii)--(iv) are then essentially identical. For (ii) we suppose there exist $K_1 < K_2$ with $C(K_1) < C(K_2)$, and pursue the strategy of selling the call with strike $K_2$ and buying the call with strike $K_1$. Then we have $g(x) = (x-K_1)_+ - (x-K_2)_+$, and so $g^*(x) = 0$, and the result follows as above. For (iii) consider either the strategy of selling the asset and buying the call with strike $0$, or selling the call with strike 0, and buying the asset. In this case we have $\In_t(S_T) = S_t$ and $\In_t(-S_T) = -S_t$, and $X_t^{\pi} = \pm(S_t-D_t^{-1} S_0)$, and the conclusion follows.

  For (iv), we have by (iii) a $K>0$ such that $C(K) < S_0 -D_TK$. We sell the asset, and buy the call with strike $K$. Then $\In_t((S_T-K)_+-S_T) = (-S_t) \vee (-D_T D_t^{-1} K) \ge -D_t^{-1} D_T K$, and so the intrinsic value of the portfolio at time $t$ is $(-C(K)+S_0) D_t^{-1} + \In_t((S_T-K)_+ - S_T) > D_t^{-1} D_T K - D_t^{-1} D_T K = 0$. The case (v) is trivial, we can buy the option for negative price, and hold to maturity.
\end{proof}

\begin{remark}
 \begin{enumerate}
 \item  Note that some of the properties of the $\In$ operator use the non-negativity of the prices process. Of course, if the asset price can go negative, then some of the conditions given can fail, e.g. in the Bachelier model.
 \item 
  Consider the case where \emph{(v)} fails. Then the `usual' arbitrage strategy would be to sell a call with a large strike, which should be worth very little, for approximately $\lim_{K \to \infty} C(K)$, and hedge in some way (or just not bother, the model-implied loss will happen with arbitrarily small probability). In our current setup, this will use up some of our lower constraint, since the intrinsic value of this strategy will remain as $-S_t$, no matter how large $K$ is. In this way, we can `use' spare $\alpha$ to generate gains, but the cost may be higher than the value of using this capacity elsewhere, depending on other elements. 
  \end{enumerate}
 \end{remark} 
  
%   \begin{question}Can we say anything else about the  point (ii)?

%   In particular, we probably want to keep this restriction (it corresponds to $\mu$ having mean $s_0$. But I don't think we ever really do e.g. embedding/martingale optimal transport type results, so maybe it isn't important.
% \end{question}

\section{Utility Maximisation in Complete Markets}

In this section we consider the problem as setup in Section~\ref{sec:setup} under the additional assumption that the trader believes in a complete market. In this case, the trader can hedge their risk, subject to the condition that the intrinsic value of their portfolio satisfies the wealth constraint, and we try to understand the impact of this on their behavior.

%\comment{
%  Outline:
%\begin{itemize}
%\item General theorem.
%\item Long position in call --- general case.
 % \begin{itemize}
%  \item  BSM. Numerical results \& Verify conditions.
%  \item Bachelier model: Specific results for calls \& Numerics
%  \item Bachelier model: optimal lambda.
%  \item Bachelier model: including One-touch, with and without calls.
 % \end{itemize}
%\item Conclusion: further work. Next challenges.
%\end{itemize}}

\subsection{Complete Market Assumption}

To make significant progress on this problem, we make the assumption that $\Pc$ is a singleton, and moreover, $\Pc = \{ \P\}$, where $\P$ is a complete market. In particular, we suppose that there exists a uniformly integrable state-price density process $H_t$, with $H_t > 0$ \as{}, $H_0 = 1$, such that whenever $Y$ is an $\Fc_T$-measurable random variable with $\Ep{D_T^{-1}H_T (1+|Y|) \log \left(1+|Y|\right)_+}<\infty$, there exists a $\mathcal{P}$-admissible portfolio $\pi$ with $\Ep{H_T Y} + X_T^\pi = Y$. In such a market, we can define as usual a risk-neutral measure $\Q$, by $\left.\frac{\di \Q}{\di \P}\right|_{\Fc_t} = H_t D_t^{-1}$. Moreover it follows from our assumptions above that $ H_t D_t^{-1}$ is a $\P$-martingale.

Additionally, we suppose that  a (set of) traded derivatives is available. To each derivative or portfolio of derivatives, we associate a fair intrinsic price process, $\In_t(C_T)$. Note that by Definition~\ref{def:intrinsic} we have $\In_t(C_T) \le C_t := H_t^{-1} \Ep{H_T C_T | \Fc_t}$, where $C_t$ is the arbitrage-free price of the derivative.

If the investor decides to take a long position in the option, her optimisation problem is:
\begin{equation}\label{eq:2}
  \text{maximise } \Ep{u(W_T^{\pi,C})}, \; \text{ subject to } \pi \in \Vc(w_0, \alpha, c_0, C_T),
\end{equation}
where $W_t^{\pi,C}$ is given by \eqref{eq:Wdefn2}; recall Definition \ref{def:admissset}.

In this section, we solve this problem for specific choices of the derivative $C$, and under a range of assumptions on the market measure $\P$. Our approach to this problem is based on results of \cite{El-Karoui:2006aa}. In that paper the authors characterise the martingale $M_t$ which maximises $\Ep{u(M_T)}$ for a concave function $u$ subject to the constraint that $M_0 = 0$ and $M_t \ge J_t$ for some supermartingale $J$. (More generally, if $J$ is not a supermartingale, it is trivial that $J$ can be replaced by its Snell envelope, \ie{} the smallest supermartingale dominating $J$). The main result of \cite{El-Karoui:2006aa} says that, if the supermartingale $J_t$ can be written in the form
\begin{equation}\label{Representation}
  J_t = \Ep{\sup_{t \le u \le T} J^*_u | \Fc_t}
\end{equation}
for some adapted process $J^*$, then the martingale $M_t := \Ep{m \vee \sup_{0 \le u \le T} J^*_u | \Fc_t}$, where $m$ is chosen so that $M_0 = 0$, maximises $\Ep{u(M_T)}$ for any concave function $u$, over the class of all martingales starting at $0$ which dominate $J$. The proof of this representation theorem can be found in \cite{El-Karoui:2008aa}.

Further, in~\cite{El-Karoui:2006aa}, the authors are able to extend their results to cover the case of utility maximisation problems under certain assumptions on the form of the function $u$. We adapt their arguments to apply to our setup. The main complication in the utility maximisation framework is that the expectation and the martingale properties of the wealth are taken under different probability measures ($\P$ and $\Q$ respectively). The aim is to put these under the same measure through an appropriate change of measure, and we require the following condition: There exists $\delta>1$ such that $\Ep{H^{-\delta}_T} <\infty$.

Following~\cite{El-Karoui:2006aa}, we suppose that our utility function is of the form $u(x) = u_p(x) = \frac{x^{1-p}}{1-p}$, where $\frac{1}{1+\delta} < p < 1$. %\comment{How necessary is this? Can we have $p > 1$? What about e.g. logarithmic utility?}

Reformulating \eqref{eq:2} using the definition of $W_t^{\pi,C}$ and the complete market characterisation of admissible dynamic trading strategies, our problem is to choose the $\Fc_T$-measurable random variable $X_T^\pi$ to maximise 
$$\Ep{u_p(D_T^{-1}(w_0 - c_0) + C_T + X_T^\pi)}$$
subject to $\Ep{H_T X_T^\pi} = 0$, $D_T^{-1}(w_0 - c_0) + C_T + X_T^\pi \ge 0$ and
\begin{equation}\label{eq:3}
   D_t^{-1}(w_0 - c_0)  + \In_t(C_T) + H_t^{-1}\Ep{H_T X_T^\pi|\Fc_t} \ge -D_t^{-1}\alpha, \quad 0 \le t < T.
\end{equation}
To write our problem in a form where we can apply the results of \cite{El-Karoui:2006aa}, we introduce
\begin{equation}\label{def:Y1}
  Y_T := D_T^{-1}(w_0 - c_0)+C_T+X_T^\pi ,
\end{equation}
so $\Ep{H_T Y_T} = w_0 + (\Ep{H_T C_T} -c_0)$ if and only if $\Ep{H_T X_T^\pi} = 0$; recall that $\Ep{D_T^{-1} H_T} =  1$. Write also $\Delta C := (\Ep{H_T C_T} -c_0) $, the difference between the hedging price (without portfolio constraints) and the market price of the option. Then, in terms of $Y$, \eqref{eq:3} becomes
\begin{equation*}
  \In_t(C_T) + H_t^{-1}\Ep{H_T Y_T|\Fc_t} - H_t^{-1}\Ep{H_T C_T | \Fc_t} \ge -D_t^{-1}\alpha ,
\end{equation*}
or equivalently, writing 
\begin{equation}\label{def:Y2}
Y_t := H_t^{-1} \Ep{H_T Y_T|\Fc_t}, 
\end{equation}
we need:
\begin{equation*}
  D_t Y_t \ge -\alpha + (D_t H_t^{-1}) \Ep{H_T C_T|\Fc_t} - D_t \In_t(C_T), \quad Y_ 0 = w_0 + \Delta C.
\end{equation*}
Note in particular that $(D_tY_t)$ is a $\Q$-martingale.

The issue now is that $DY$ is a martingale under $\Q$, while we maximise under the probability $\P$. To get around this difficulty we introduce a new measure $\bQ$ under which the two conditions can be understood for the same measure. For this purpose, we define $(\xi_{\cdot})$ by
\begin{equation}
  \label{eq:4}
  \xi_T := \frac{H_T^{-\frac{1}{p}}}{\Ep{H_T^{1-\frac{1}{p}}}}, \quad \xi_t := H_t^{-1} \Ep{\xi_T H_T | \Fc_t}.
\end{equation}
Then $\xi_0 = 1$,  $\xi_t H_t$ is a $\P$-martingale, and we can define a change of measure $\frac{\di \bQ}{\di \P} = \xi_T H_T$.

\begin{lemma} \label{lem:Ybar}
  Under $\bQ$, with utility function $u = u_p$, the problem \eqref{eq:2} is equivalent to the problem:
  $$
   \text{maximise } \Ebq{u_p(\ol{Y}_T)},
  $$
  subject to
  \begin{equation}\label{eq:5}
  \begin{cases}
    \ol{Y} &\text{ is a non-negative $\bQ$-martingale}, \\
     \ol{Y}_0 &=  w_0 + \Delta C,\\%\label{eq:5}\\
     \ol{Y}_t &\ge  -\alpha \xi_t^{-1} D_t^{-1} -\xi_t^{-1} \In_t(C_T)  + \Ebq{C_T \xi_T^{-1}| \Fc_t}.
    \end{cases}
  \end{equation}
Moreover, the optimal terminal wealth $W_T^{\pi,C}$ can be recovered by $W_T^{\pi,C} = \xi_T \ol{Y}_T$, where $\ol{Y}_T$ is the optimiser for the problem above.
\end{lemma}

\begin{proof}
  Let $X_T^\pi$ be a candidate solution to \eqref{eq:2}, so $\Ep{H_T X_T^\pi} = 0$ and $W_T^{\pi,C} = D_T^{-1}(w_0 - c_0) + C_T + X_T^\pi \ge 0$. Then, from (\ref{def:Y1}, 
  \begin{equation*}
    \Ep{u_p(D_T^{-1}(w_0 - c_0) + C_T + X_T^\pi)} = \Ep{u_p(Y_T)}.
  \end{equation*}
  Define $\ol{Y}_t = Y_t \xi_t^{-1}$; see (\ref{def:Y2}). Then $\Ebq{\ol{Y}_T | \Fc_t} = (\xi_t H_t)^{-1}\Ep{Y_T \xi^{-1}_T H_T \xi_T| \Fc_t} =Y_t  \xi_t^{-1} = \ol{Y}_t$. In particular, $\ol{Y}$ is a $\bQ$-martingale, and $\ol{Y}_T \ge 0$, hence $\ol{Y}$ is a non-negative martingale.
  
  On the other hand, we have 
  \begin{align*}
    \Ep{u_p(Y_T)} 
    & = \Ep{\frac{Y_T^{1-p}}{1-p}}\\
    & = \Ebq{(H_T \xi_T)^{-1} \frac{(\ol{Y}_T\xi_T)^{1-p}}{1-p}}\\
    & = \Ebq{\frac{\xi_T^{-p}}{H_T} u_p(\ol{Y}_T)}\\
    & = \left(\Ep{H_T^{1-\frac{1}{p}}}\right)^{p} \Ebq{u_p(\ol{Y}_T)}.
  \end{align*}
  Moreover, the constraints on the lower bound and the initial value of $\ol{Y}$ follow immediately. As a consequence, any feasible solution $Y$ to the first problem gives rise to a feasible solution $\ol{Y}$ to the second problem, and the corresponding values differ only by a positive constant multiple.  A similar conclusion can be obtained starting with a feasible solution $\ol{Y}$  to  \eqref{eq:5}, and building a candidate solution to   \eqref{eq:2}. Since the market is complete,  the arbitrage free price and the replication strategy of any contingent claim are uniquely determined, and hence there exists an admissible dynamic trading strategy $\pi$  such that  the $\Q$-martingale  $D_t X_t^{\pi}$ replicates the contingent claim $D_T \xi_T \ol{Y}_T-D_T C_T-(w_0-c_0)$.   
\end{proof}

\begin{remark}\label{rem:restprop}
  Note that the intrinsic wealth constraint process
  \begin{equation} \label{eq:ZetaDef}
    \zeta_t := -\alpha D_t^{-1}\xi_t^{-1} - \xi_t^{-1} \In_t(C_T) + \Ebq{ \xi_T^{-1} C_T | \Fc_t}
  \end{equation}
  is a $\bQ$-supermartingale.  To see this, we observe that $\frac{\di \bQ}{\di \Q} = D_T \xi_T$, so $(D_t^{-1}\xi_t^{-1})$ is a $\bQ$-martingale, and $(D_t \In_t(C_T))$ is a $\Q$-submartingale by the assumption that $\In_t(C_T)$ is a fair intrinsic wealth process, so also $(\xi_t^{-1} \In_t(C_T))$ is a $\bQ$-submartingale.
  % \comment{Is the discount factor on $\alpha$ correct? Wrong initial assumption?}
  
\end{remark}

  More generally, we can reformulate the conditions in Lemma~\ref{lem:Ybar} that $(\ol{Y})_{t \le T}$ is a non-negative $\bQ$-martingale such that $\ol{Y}_t \ge \zeta_t$ for $t \le T$ in terms of the process
  \[
    \zeta_t^0 :=
    \begin{cases}
      \zeta_t & t < T\\
      0 & t = T
    \end{cases}.
  \]
  Specifically, since $\zeta_T = -\alpha D_T^{-1}\xi_T^{-1} < 0$, the requirement that $\ol{Y}$ is non-negative and greater than $\zeta_t$ is equivalent to requiring that the process $\ol{Y}$ is greater than $\zeta^0$, and further, equivalently, that it is greater than the Snell envelope of $\zeta^0$, which we denote by $\zeta^*$.

%\cred{Do we now need the statement below for the process $\zeta^*$? Check, maybe the assumption can be made on $\zeta$, not its Snell envelope?}

\begin{assumption} \label{ass:sup-decomp}
 %Given the Meyer $\sigma$-field  $\Lambda$, 
 The process $\zeta^*$ is a supermartingale of class  $({\mathcal D})$ upper semi-continuous in expectation which admits the decomposition in terms of an optional, upper-right semi-continuous process $J_u^\zeta $, with $J_T^\zeta =0$, as 
 \begin{equation}
    \label{eq:7}
    \zeta_t^* = \Ebq{\sup_{t \le u \le T} J_u^\zeta | \Fc_t}.
  \end{equation}   
 \end{assumption}
% Then,
%\begin{equation}
 %   \label{eq:7}
 %   \zeta_t^* = \Ebq{\sup_{t \le u \le T} J_u^\zeta | \Fc_t},
  %\end{equation}
% implying that   
%\begin{assumption} \label{ass:sup-decomp}
  %We suppose that the 
 % process  $\zeta^*$ can be decomposed in terms of an optional, right-upper semi-continuous process $J^\zeta$.
  % such that:
  %\begin{equation}
    %\label{eq:7}
   % \zeta_t = \Ebq{\sup_{t \le u \le T} J_u^\zeta | \Fc_t}.
  %\end{equation}
  By \cite[Theorem~2.9]{bank_stochastic_2021}, we have that the representation in (\ref{eq:7}) can be obtained for process $\zeta^*$. Observe that, following \cite{El-Karoui:2006aa} and \cite{{El-Karoui:2008aa}}, in order to obtain this representation 
  %can obtained  is a mild assumption. (In fact, according to \cite[Theorem~2.7]{El-Karoui:2008aa}, 
  it is sufficient that the filtration $ \{\Fc_t\}$  is quasi-left-continuous.
  %, and that $\zeta^*$ is a quasi-left-continuous supermartingale of class $(\mathcal{D})$).
%\end{assumption}

\begin{theorem} \label{thm:main_decomp}
  Suppose that $u = u_p$, $\In_t(C_T)$ is a fair intrinsic price process and Assumption~\ref{ass:sup-decomp} holds. Then there exists a feasible solution to \eqref{eq:2} if and only if
  % \begin{equation*}
  %   w_0 -c_0 +\alpha + \In_0(C_T) \ge 0 \text{ and } \Ebq{\left(
  %       \sup_{0 \le u \le T} J_u^\zeta\right) \vee 0} \le w_0 + \Delta C .
  % \end{equation*}
  \begin{equation*}
    \Ebq{\sup_{0 \le u \le T} J_u^\zeta} \le w_0 + \Delta C .
  \end{equation*}
  When this condition holds, the optimal terminal wealth $W_T^\pi$ solving \eqref{eq:2} is given by
  \begin{equation*}
    W_T^\pi = \xi_T\left[ \left( \sup_{0 \le u \le T} J_u^\zeta\right) \vee M\right],
  \end{equation*}
  where $M$ is chosen such that $\Ebq{\left( \sup_{0 \le u \le T} J_u^\zeta\right) \vee M} = w_0 + \Delta C$.
\end{theorem}

\begin{proof}
  Under Assumption~\ref{ass:sup-decomp}, it is immediate that $\ol{\zeta}_t := \Ebq{\sup_{0 \le u \le T} J_u^\zeta|\Fc_t}$ is a non-negative martingale which dominates $\zeta^*$. Moreover, it is the smallest such martingale. 
  
  In particular, if $\ol{\zeta}_0  \le w_0 + \Delta C$, then there exists a process $\ol{Y}$ which is feasible for \eqref{eq:5}. Moreover, since any admissible wealth process gives rise (via the arguments of Lemma~\ref{lem:Ybar}) to a non-negative martingale which dominates $\zeta^*$, this is also a necessary condition. The result now follows immediately from \cite[Theorem~5.2]{El-Karoui:2008aa} \footnote{Note that there is a typographical error in the first bullet point on p.685 of \cite{El-Karoui:2008aa}, and this point is not in general correct: in our notation, the choice of $M$ will in general be less than $w_0 + \Delta C$; we do not usually expect equality except in special cases.}
\end{proof}

It is now straightforward to deduce the form of the optimal $\pi$, using classical methods, see for example~\cite[Theorem~6.3, Corollary~6.5]{karatzas_methods_1998}.

\subsection{Long position in Call options}

In this section we consider the case where the derivative position is a long position in call options. Specifically, we suppose that the agent purchased $\lambda > 0$ units of a call option with strike $K$. As we will see, in this case the form of the optimal terminal wealth can be identified.

We first consider the fair (replication) price of a call option with strike $K$ at time $t$, given by:
\begin{align*}
 H_t^{-1} \Ep{H_T (S_T-K)_+|\Fc_t} 
  & = \xi_t \Ebq{\xi_T^{-1} (S_T-K)_+|\Fc_t} \\ 
  & = D_t^{-1} \Eq{D_T (S_T-K)_+|\Fc_t}.
\end{align*} 
%\comment{Do we use this definition?}
In addition, from Example~\ref{ex:Intrinsic}{\it(\ref{item:1})}, we know that $\In_t(\lambda(S_T-K)_+) = \lambda \left(S_t - K \frac{D_T}{D_t}\right)_+$, and so:
\begin{align*}
  \zeta_t = D_t^{-1} \xi_t^{-1}   \left[ - \alpha-\lambda (S_t D_t - K D_T)_+ +\lambda \Eq{(S_T D_T -KD_T )_+| \Fc_t}    \right].
\end{align*}
Observing that $S_t^D:= D_t S_t$ is a $\Q$-martingale, and writing $K^D :=  D_T K$, it follows that
\begin{equation}
  \label{eq:8}
  \zeta_t =\lambda  \xi_t^{-1} D_t^{-1} \left(\Eq{L_{T}^{S^D,K^D}-L_t^{S^D,K^D}|\Fc_t}-\frac{\alpha}{\lambda}\right)
\end{equation}
where $L^{S^D,K^D}$ is the local time of the process $S^D$ at the level $K^D$. Recalling again that $\frac{\di \bQ}{\di \Q} = D_T \xi_T$, we can also write
\begin{equation*}
    \zeta_t = \Ebq{\left(\lambda \left(L_{T}^{S^D,K^D}-L_t^{S^D,K^D}\right)-\alpha\right) \xi_T^{-1} D_T^{-1} |\Fc_t}.
\end{equation*}

Let us introduce $\phi_t := \xi_t^{-1} D_t^{-1}$, so we can write 
\begin{align*}
  \left(L_{T}^{S^D,K^D}-L_t^{S^D,K^D}\right) \phi_T
  & = \left(L_{T}^{S^D,K^D}-L_t^{S^D,K^D}\right) \phi_t  + \int_t^T  \di L_s^{S^D,K^D} \int_t^T \di \phi_s
\end{align*} 
and by integration by parts,
\begin{align*}
  & \int_t^T  \di L_s^{S^D,K^D} \int_t^T \di \phi_s\\
  & \quad = \int_t^T (\phi_s-\phi_t) \,\di L_s^{S^D,K^D} + \int_t^T \left(L_{s}^{S^D,K^D}-L_t^{S^D,K^D}\right)\, \di\phi_s\\
  & \quad = \int_t^T\phi_s\, \di L_s^{S^D,K^D} - \left(L_{T}^{S^D,K^D}-L_t^{S^D,K^D}\right) \phi_t
    + \int_t^T \left(L_{s}^{S^D,K^D}-L_t^{S^D,K^D}\right) \,\di \phi_s
\end{align*}
and so 
\begin{align*}
  \left(L_{T}^{S^D,K^D}-L_t^{S^D,K^D}\right) \phi_T
  & = \int_t^T \phi_s\, \di L_s^{S^D,K^D} + \int_t^T \left(L_{s}^{S^D,K^D}-L_t^{S^D,K^D}\right) \,\di\phi_s.
\end{align*}
Since $\phi$ is a $\bQ$-martingale, it follows that 
\begin{align*}
    \Ebq{\left(L_{T}^{S^D,K^D}-L_t^{S^D,K^D}\right)\phi_T |\Fc_t} = \Ebq{\int_t^T\phi_s \di L_s^{S^D,K^D}|\Fc_t}.  
\end{align*}

To say more about the optimal strategy in this framework, we then make the following  assumption:
\begin{assumption} \label{ass:barrier_markov}
  Suppose that $S$ is a time-homogenous, Markov processes under $\bQ$, and there exists a measurable function $\phi(u)$ such that $\phi_u = \phi(u)$ when $S_u^D = K^D$.  That is, $\xi_tD_t$ does not depend on the past of the process when the discounted price and the discounted strike are equal.
\end{assumption}
Below we will see that this assumption holds for the case of a
Black-Scholes-Merton model. In particular, it follows from Assumption~\ref{ass:barrier_markov} that
\begin{align*}
  \rho(t) :=
  \begin{cases}
    \Ebq{\int_t^T \phi_s \, \di L_s^{S^D,K^D} | S_t^D = K^D}, \quad & t < T\\
    -\infty, & t \ge T
  \end{cases}
\end{align*}
is well defined.

%{DELETE: \cblue
%Let
%\begin{align*}
 % \psi_t(m):= \Ebq{m \wedge \phi_T| S_t^D = K^D},
%\end{align*}
%and note that $\psi_t(m)$ is strictly increasing in $m$ on the set $\left\{m\big|
%\Pr[\phi_T < m|S_t^D = K^D] > 0\right\}$. 
%}

%Then we have the following result:
\begin{theorem} \label{thm:long-position-call}
  Suppose that Assumption~\ref{ass:barrier_markov} holds, and in addition, $\phi$ is a Markov process and
  \begin{align}\label{def:z}
    z(u; \lambda) := \lambda \rho(u) - \alpha \phi(u)
  \end{align}
  is decreasing in $u$, for $u \in [0,T]$. Then the process  $J^\zeta$ involved in the representation of $\zeta^*$ in (\ref{eq:7}),  the Snell envelope of $\zeta^0$,
  is given by 
   \begin{align*}
    J_u^\zeta =
    \begin{cases}
      z(u; \lambda) & S_u^D = K^D\\
      0 & u = T\\
      - \infty & \text{ otherwise},
    \end{cases}
  \end{align*}
  and so
  \begin{align*}
    \sup_{0 \le u \le T} J_u^\zeta =
    \begin{cases}
      z(H_{K^D}; \lambda) \vee 0 & \quad H_{K^D} < T\\
      0 & \quad T \wedge H_{K^D} = T.
    \end{cases}
  \end{align*}
\end{theorem}

\begin{proof}
  
 Recall that the process $\zeta$ is defined in \eqref{eq:ZetaDef}, and   consider the process $J^\zeta$ defined above.
We first note that the process $\zeta$ can be described as follows.  Considering first $t=T$, we see immediately that $\zeta_T = - \alpha\phi_T$. In addition, for $t < T$, define  the  stopping time
\begin{equation}\label{def:stoppingtime}
H_{K^D}^t := \inf \{s \ge t: S_s^D = K^D\},
\end{equation}
then $(\zeta_{s \wedge H_{K^D}^t})_{s \in [t,T]}$ is a martingale, and $\zeta_t = \Ebq{\zeta_{T \wedge H_{K^D}^t}}$. 
Now, observe that $\zeta_{H_{K^D}} = z(H_{K^D}; \lambda)$ when $H_{K^D} < T$. 
Since $\zeta^*_{H_{K^D}} \ge \zeta_{H_{K^D}}$, we have that, 
  $$\zeta^*_{H_{K^D}} \ge \Ebq{\sup_{{H_{K^D}} \le u \le T} J_u^\zeta | \Fc_{H_{K^D}}} = \zeta_{{H_{K^D}}} \vee 0 =z(H_{K^D};\lambda)\vee 0, 
  $$
when $H_{K^D} < T$.  It follows that 
   $$\zeta^*_0 \ge \Ebq{\left(\zeta^*_{H_{K^D}}\vee 0 \right)\ind_{H_{K^D} < T}} \ge \Ebq{\sup_{0\le u \le T} J_u^\zeta}.$$
 
    On the other hand, by construction, $\Ebq{\sup_{t \le u \le T} J_u^\zeta | \Fc_t}$ is a non-negative supermartingale dominating $\zeta_t$, from which we conclude
    \[
      \zeta_t^* = \Ebq{\sup_{t \le u \le T} J_u^\zeta | \Fc_t},
    \]
    as required.

  % , we would therefore expect on $S_t^D = K^D$ that $\zeta_t \ge \Ebq{J_t^\zeta \vee (-\alpha \phi_T)|S_t^D = K^D}$, with equality if $J_t^\zeta$ is decreasing on $\{S_t = K\}$. In particular, on $S_t^D = K^D$ we have
  % \begin{equation*}
  %   \zeta_t = \rho(t) - \alpha \phi(t) \ge \Ebq{J_t^\zeta \vee (-\alpha \phi_T)|S_t^D = K^D} = -\alpha\psi_t(-z(t)/\alpha).    
  % \end{equation*}
  % When $z$ is as given, it follows that there is equality, and the result is shown.

  % $\Ebq{J_t^\zeta \vee (-\alpha \phi_T)|S_t^D = K^D} = -\alpha\psi_t(\phi(t)-\alpha^{-1} \rho(t))$ as required. The conclusion follows under the assumption that this is decreasing.
  %
\end{proof} 

\begin{remark} The assumption that $z$ is decreasing is necessary to get an explicit expression. In general, we would expect the process $J^\zeta$ to take a similar form, but it would no longer be the case that $\sup_{0 \le u \le T} J_u^\zeta = z(H_{K^D}; \lambda) \vee 0$, and rather, the right-hand side would be a maximum over all possible return times to the level $K^D$. In this case, we would expect the function $z$ at time $u$ to then only be defined recursively in terms of an expression involving its future values, $\{z(s; \lambda), s \in (u,T]\}$.
\end{remark}

Next result is in fact a special case of the previous more general
result.
\begin{corollary} \label{lem:rho-dec}
  Suppose that Assumption~\ref{ass:barrier_markov} holds, $\alpha = 0$, and the function $\rho(t)$ is strictly positive and decreasing. Then, we have
  \begin{align*}
    \sup_{0 \le u \le T} J_u^\zeta = \lambda \rho(H_{K^D}) \vee 0
  \end{align*}
  where $H_{K^D} := \inf \{ t \ge 0 : S_u^D = K^D\}$.
\end{corollary}
\begin{proof}
 Taking $z(u,\lambda)=\lambda \rho (u)$, and using the fact that $\rho$ is decreasing and strictly positive, the result follows from the expression for $J^\zeta_u$, taking $\alpha=0$, since the maximum of this process within the interval $[t,T]$ will be achieved at the valuation in the  left limit $t$.
\end{proof}

\subsection{Long positions in Call options in the Black-Scholes-Merton model}

We now restrict ourselves to the standard setting of the Black-Scholes-Merton model, so that $D_t = \me^{-rt}$, $\di S_t = \sigma S_t \di B_t + \mu S_t \dt$, for fixed constants $\sigma, r, \mu$, where $B_t$ is a $\P$-Brownian motion. In this model, we know that \[H_t = \exp\left\{-\theta B_t - \left(\half \theta^2 +r\right) t \right\},\] where $\theta = \frac{\mu-r}{\sigma}$, the Sharpe ratio.

Using the fact that $\xi_T = H_T^{-\frac{1}{p}} \left(\Ep{H_T^{1-\frac{1}{p}}}\right)^{-1}$, and $H_t \xi_t$ is a $\P$-martingale, one can further see that
\begin{equation*}
  \xi_t = \exp\left\{ \frac{\theta}{p} B_t + rt + \half \theta^2 (2 p^{-1} - p^{-2}) t\right\},
\end{equation*}
and therefore
\begin{align}
  \phi_t^{-1} = D_t\xi_t = \exp\left\{ \frac{\theta}{p} B_t^{\Q} - \half \frac{\theta^2}{p^2}t\right\}, \label{eq:11}
\end{align}
where $B_t^\Q := B_t + \theta t$ is a $\Q$-Brownian motion. It follows that $S_t^D = K^D$ if and only if:
\begin{align}
  K^D & = S_0 \exp\left\{ \sigma B_t^\Q - \half \sigma^2 t\right\} \nonumber\\
  \iff B_t^\Q & = \frac{1}{\sigma} \left[ \ln \left( \frac{K^D}{S_0}\right) + \half \sigma^2 t\right]       \label{eq:10}
\end{align}
and hence
\begin{align*}
  \phi(t) & = \exp\left\{ -\frac{\theta}{p} \frac{1}{\sigma} \left[ \ln \left( \frac{K^D}{S_0}\right) + \half \sigma^2 t\right]  + \half \frac{\theta^2}{p^2}t\right\}\\
  & = \left( \frac{S_0}{K^D}\right)^{\frac{\theta}{\sigma p}} \exp\left\{ \frac{\theta}{2p^2} \left( \theta-\sigma p\right)t\right\}.
\end{align*}
Note that it follows that Assumption~\ref{ass:barrier_markov} holds, and moreover $\phi$ is decreasing in $t$ if $\theta>0$ and $p \sigma > \theta$.

%\comment{Is it also true that $z$ is decreasing? Can we give conditions?}

Using from above the fact that  $\frac{\di \bQ}{\di \Q} = D_T \xi_T$, which is given by \eqref{eq:11}, we also see from Girsanov's Theorem that $B_t^{\bQ} := B_t^\Q - \frac{\theta}{p} t$ is a $\bQ$-Brownian motion.

In addition, using the Black-Scholes formula, we have 
\begin{equation*}
  \Eq{L_{T}^{S^D,K^D}-L_t^{S^D,K^D}|\Fc_t} = \left( \Phi(d_1) S_t^D - \Phi(d_1-\sigma\sqrt{T-t})K^D\right) - (S_t^D-K^D)_+
\end{equation*}
where $d_1 = \left( \log\left( \frac{S_t^D}{K^D}\right) + \sigma^2 \sqrt{T-t}/2\right) /(\sigma \sqrt{T-t})$, and it follows from the argument used to derive \eqref{eq:8} that
\begin{equation*}
  \rho(t)  = \phi(t) K^D \left( 2 \Phi\left(\half \sigma \sqrt{T-t} \right) -1\right).
\end{equation*}
From the fact that $\phi$ is decreasing, we deduce that Theorem~\ref{thm:long-position-call} 
 holds when $\theta>0$ and $p \sigma > \theta$, since  in this case $z(\cdot;\lambda)$, defined by (\ref{def:z}), is decreasing. This last fact follows from the following observation, using the previous display.  
 \begin{align*}
 z(u;\lambda)=&\lambda\rho(u)-\alpha \phi(u)\\
 =&-\alpha\phi(u)\left[1-\frac{\lambda K^D}{\alpha}\left(2\Phi(\frac{1}{2}\sigma\sqrt{T-u})-1\right)\right]\\
 =:&-\alpha\phi(u)g(u),
 \end{align*}
with $g(T)=1$, $g(0)<1$ and $g'(u)>0$, for $0<u<T$.  The functions $\rho, \phi$ and $z$ are shown in Figures~\ref{fig:RhoPhi} and \ref{fig:zPlot}.

%\comment{And z decreasing!}

\begin{figure}
  \includegraphics[width=\textwidth]{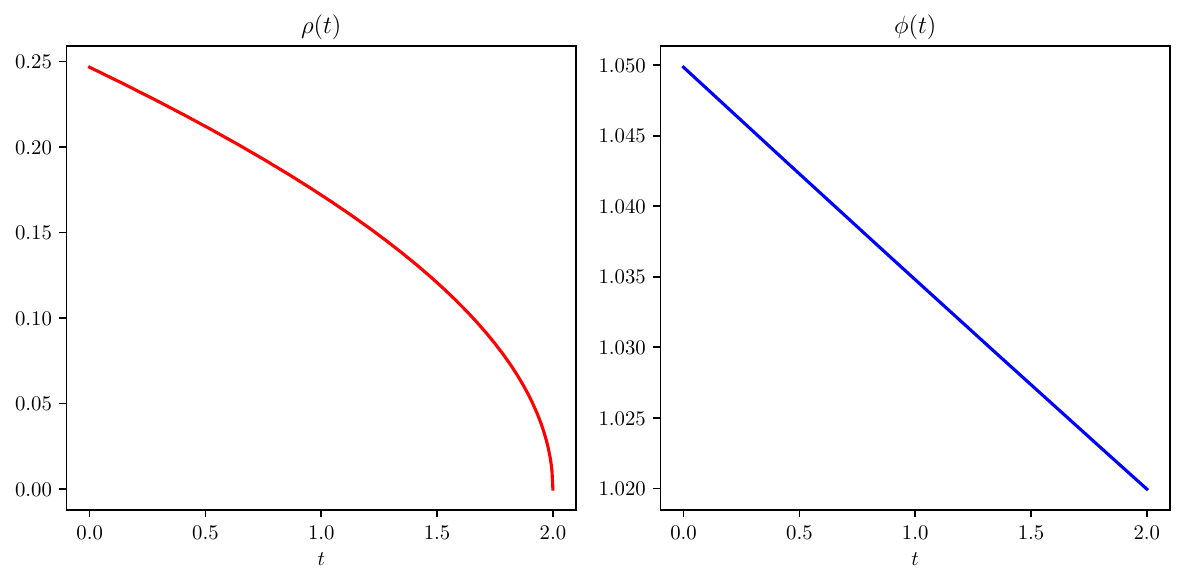}
  \caption{\label{fig:RhoPhi} Plots of the functions $\rho$ and $\phi$ in the case of the Black-Scholes-Merton model. In this example, $\rho$ represents the additional cost of hedging that needs to be held when $S_t^D = K^D$, $\phi$ is the value of the process $\phi_t$ along the same line.}
\end{figure}

\begin{figure}
  \includegraphics[width=\textwidth]{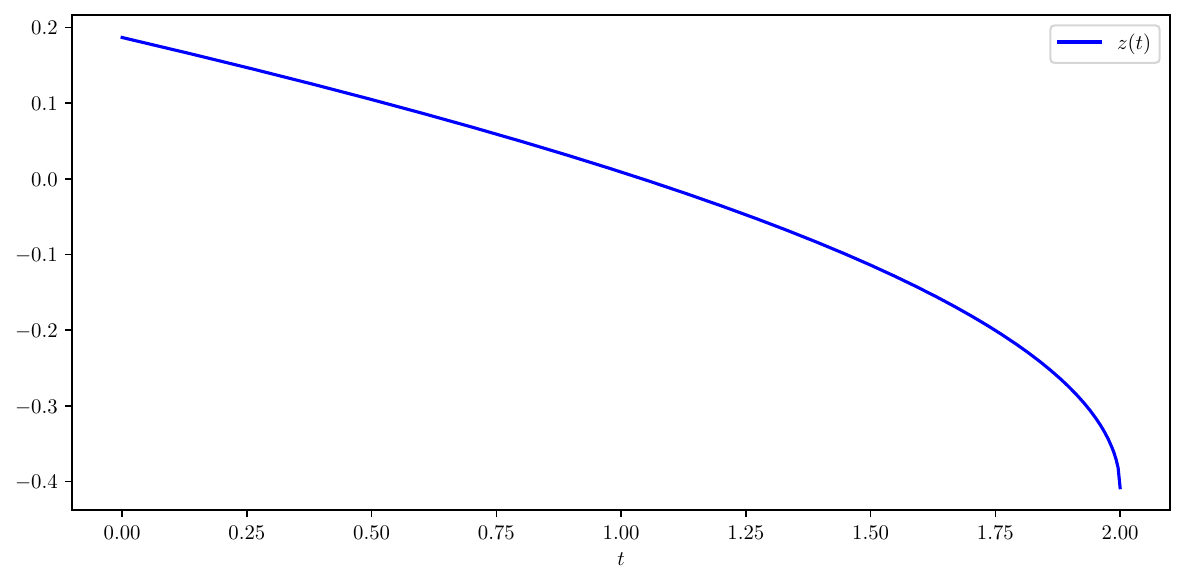}
  \caption{\label{fig:zPlot} The plot above shows the function $z(t,\lambda)$. The level $z$ represents the value of the process $J^\zeta$ which we will score when we hit the line $S_t^D = K^D$ in order to get the correct form of the process defined in \eqref{eq:7}.}
\end{figure}

For $y \neq x$, $\beta \in \R$, and for a standard Brownian motion $B$ with $B_0 = 0$, we introduce the hitting time
\begin{align*}
  H_{y}^{\beta} := \inf \{ t \ge 0 :  x+ B_t = y + \beta t\}
\end{align*}
and define the densities $\gamma_0, \gamma_1^\beta, \gamma_2^\beta$ by:
\begin{align}
  \P(x+ B_t \in A) & = \int_A \gamma_0(v,t,x) \, \di v \nonumber\\
  \label{eq:BrownianHittingLine}
  \P(H_y^\beta < t) & = \int_0^t \gamma_1^\beta(u,x,y) \, \di u\\
  \P(H_y^\beta > t, x+ B_t \in A)  & = \int_A \gamma_2^\beta(v,t,x,y) \, \di v
                                  \label{eq:BrownianNotHitting}
\end{align}

\begin{proposition} \label{prop:BSM_M}
  In the Black-Scholes-Merton model with $\theta >0, p \sigma > \theta$, $z$ decreasing and $S_0 \neq K^D$, we can write 
  \begin{align}
   \Ebq{\left( \sup_{0 \le u \le T} J_u^\zeta\right) \vee M}
    & = M + \int_0^T \gamma_1^\beta(u,x,y)  (z(u; \lambda) - M)_+  \, \di u \nonumber
    \\
    & = M + \int_0^{r^*(M;\lambda)} \gamma_1^\beta(u,x,y)  (z(u; \lambda) - M) \, \di u
    % & \qquad {} + \int_\R  \gamma_2^\beta(v,T,x,y) \left( \tilde{\varphi}_T(v) \wedge \left(-\frac{M}{\alpha}\right) \right) \, \di v
    \label{eq:12}
  \end{align}
  where $x = \sigma^{-1} \log(S_0), y = \sigma^{-1} \log(K^D), \beta = \sigma/2-\theta/p$, and
  \begin{equation*}
    r^*(M;\lambda) := \inf\{ u < T : z(u; \lambda) < M\} \wedge T.
  \end{equation*}
  % \begin{equation*}
  %   \tilde{\varphi}_T(u) = \exp\left\{ -\frac{\theta}{p} u - \frac{1}{2} \frac{\theta^2}{p^2} T\right\}.
  % \end{equation*}
\end{proposition}

\begin{proof}
  First note that the parameters $x, y$ and $\beta$ come from the relevant terms in \eqref{eq:10} together with $B_t^{\bQ} := B_t^\Q - \frac{\theta}{p} t$.
  
  From Theorem~\ref{thm:long-position-call} we observe that
  \begin{align*}
    \sup_{0 \le u \le T} J_u^\zeta = z(H_{K^D};\lambda) \vee 0
  \end{align*}
  holds, and therefore 
  \begin{align*}
    \sup_{0 \le u \le T} J_u^\zeta \vee M = M + (z(H_{K^D};\lambda) - M)_+.
  \end{align*}
  The result now follows upon noting that $z$ is decreasing and the observation that $H_{K^D}$ has distribution given by $\gamma_1^\beta$.
\end{proof}

% In the following lemma we give an alternative characterisation of the optimal choice of $M$, which will enable the use of relatively simple numerical methods for computing $M$.

We now put together the results of this section to give a complete characterisation of the optimal wealth in the case of a long position in call options.

\begin{theorem}
  Suppose the conditions of Proposition~\ref{prop:BSM_M} and Theorem~\ref{thm:main_decomp} hold.
  Then there exists an admissible trading strategy which is long $\lambda$ units of the Call option with strike $K$ if and only if
  \[
    w_0 + \lambda \Delta C  \ge \int_0^{r^*(0;\lambda)} \gamma_1^\beta(u,x,y)  z(u; \lambda) \, \di u.
  \]
  If this holds then the value of $M$ in the trader's optimal portfolio is the unique solution to the equation
  \begin{align}
    w_0 + \lambda \Delta C & = M + \int _0^{r^*(M;\lambda)} \gamma_1^\beta(u,x,y) \left\{z(u;\lambda)-M\right\} \, \di u, \label{eq:M-BSM-Simplified}
  \end{align}
  and the optimal utility is given by
  \begin{equation*}
     \Ep{u_p(W_T^{\pi,C})} = c_p \cdot \left( u_p(M) + \int _0^{r^*(M;\lambda)} \gamma_1^\beta(s,x,y) \left\{u_p(z(s;\lambda))-u_p(M)\right\} \, \di s\right),
  \end{equation*}
  where $M:= M(\lambda)$ is the value given by~\eqref{eq:M-BSM-Simplified} and $c_p = \left(\Ep{H_T^{1-\frac{1}{p}}}\right)^{p}$.

\end{theorem}

% \comment{Can we be more specific here? What about if $\theta=p\sigma/2$? Then we should be able to compute the hitting density of the line, and compute many things numerically at least.}

% Now suppose in addition that \fbox{$\theta = p\sigma/2$}. Then we have
% \begin{equation*}
%   S_t^D = S_0 \exp\left\{ \sigma B_t^\Q - \half \sigma^2 t\right\} = S_0 \exp\left\{ \sigma B_t^{\bQ} + \left( \frac{\theta \sigma}{p} - \half \sigma^2\right) t\right\}.
% \end{equation*}
% In particular, the first hitting time (under $\bQ$) of $K^D$ by $S^D$ is equal to the first hitting time of $\frac{\ln(K^D/S_0)}{\sigma}$ by $B_t^{\bQ}$. This latter term has nice expressions, so we should be able to do things here.

\begin{proof}
  The first claim follows from Proposition~\ref{prop:BSM_M} and Theorem~\ref{thm:main_decomp}. It also follows from this and the fact that with positive probability, since $S_0 \neq K^D$, $H_{K^D} \ge T$ that a unique value of $M$ satisfying \eqref{eq:M-BSM-Simplified} exists.
  
  The form of the optimal utility now follows from applying the known distribution of $\left( \sup_{0 \le u \le T} J_u^\zeta\right) \vee M$, and the arguments of Lemma~\ref{lem:Ybar}.
\end{proof}

\begin{remark}
  In general the expression for $\gamma_1^\beta$ does not exist in closed form (unless $\beta = 0$), however it is well known that its Laplace transform can be given in closed form. In combination with the fact that the right hand side of \eqref{eq:M-BSM-Simplified} is increasing in $M$, this means that the optimal value of $M$ can be found quickly via simple numerical methods.
\end{remark}

\begin{remark}
  In the case where the initial value of the asset is equal to the at-the-money strike, $S_0 = K^D$, then $H_{K^D} \equiv 0$ and so the resulting $\sup_{0 \le u \le T} J^\zeta_u$ is fixed. Hence the optimal wealth will also be deterministic, and there hence this case is trivial. 
\end{remark}

\subsection{Numerical Results}

In the context of the above results, it is possible now to numerically compute various relevant quantities to get a sense of the typical behaviour. We show the results of such numerical computations below.

\begin{figure}
  \includegraphics[width=\textwidth]{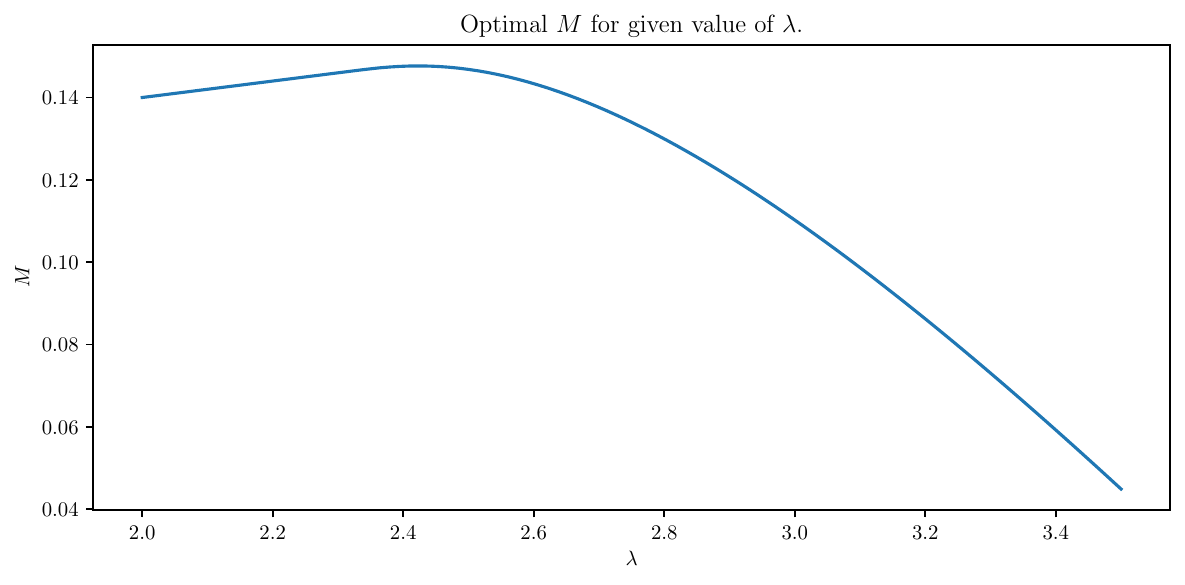}
  \caption{\label{fig:MPlot} The figure shows the value of $M$ as a function of $\lambda$. The parameters used are $S_0 = 1.2, K = 0.85, T=2, \sigma = 0.5, r = 0.01, \alpha = 0.4, \Delta C = 0.02, p=0.75, \theta = 0.05$.}
\end{figure}

\begin{figure}
  \includegraphics[width=\textwidth]{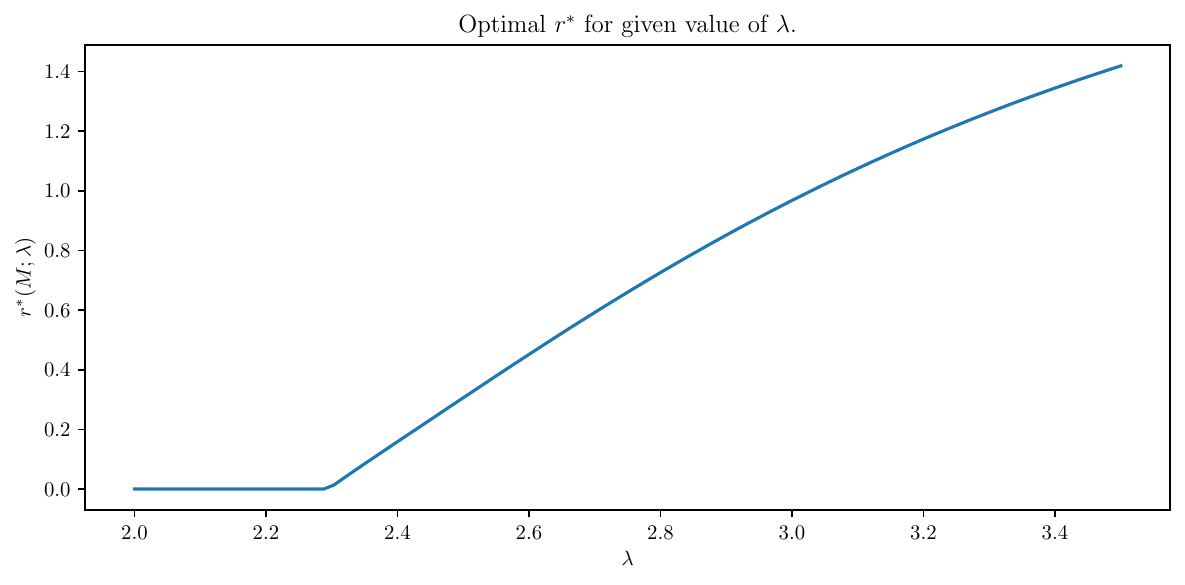}
  \caption{\label{fig:RStarPlot} The figure shows the value of $r^*(M(\lambda),\lambda)$ as a function of $\lambda$.  The parameters are the same as in Figure~\ref{fig:MPlot}. Note that for $\lambda$ small, then $M(\lambda) \ge z(0;\lambda)$, the intermediate wealth constraint is never binding, and $r^*(M,\lambda) = 0$.}
\end{figure}

\begin{figure}
  \includegraphics[width=\textwidth]{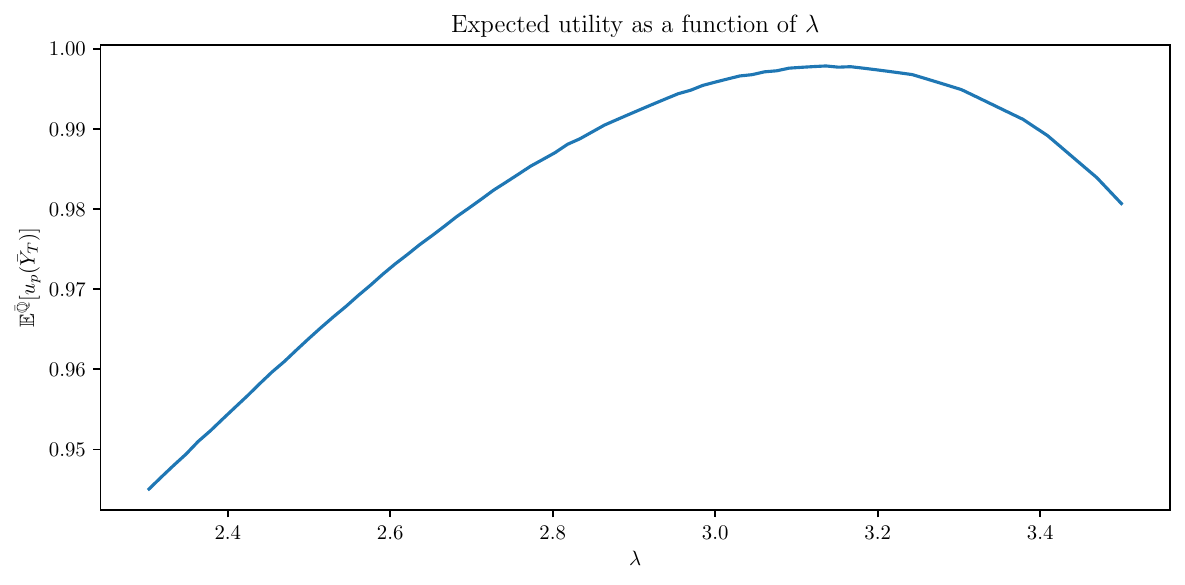}
  \caption{\label{fig:UtilityPlot} The figure shows the value of the utility as a function of $\lambda$. The parameters are the same as in Figure~\ref{fig:MPlot}. In this case, the optimal utility occurs when $\lambda \approx 3.1$.}
\end{figure}

We observe that $M$ appears to be concave in $\lambda$, and we see similar behaviour when we plot the utility, where the optimal value utility is attained for $\lambda$ approximately equal to $3.1$. The value of $r^*(M,\lambda)$ is increasing in $\lambda$, although there is an initial interval where it is equal to zero, since $M(\lambda) \ge z(0;\lambda)$, and the intermediate wealth constraint is never binding. Note that this interval also corresponds to linear behaviour for the value of $M$ as a function of $\lambda$, since here we can directly hedge the exposure, and we always guarantee a terminal wealth equal to $w_0 + \lambda \Delta C \equiv M$.

The explanation for this behaviour can be seen in Figure~\ref{fig:CDF_YT}, which shows the cumulative distribution function of the optimal wealth $Y_T$ for the optimal strategy. We see that as we increase $\lambda$, the mean of the distribution increases, but the variance also increases. The large jump on the left of each distribution function corresponds to the value of $M$, and for small values of $\lambda$, then this step corresponds to a large proportion of the distribution. As $\lambda$ increases, the proportion of the distribution that is at $M$ decreases, and the distribution becomes more spread out. As we increase $\lambda$, the mean of the distribution also increases. The optimal choice of $\lambda$ is then determined through a trade-off between the mean and variability of the resulting distribution.

\begin{figure}
  \includegraphics[width=\textwidth]{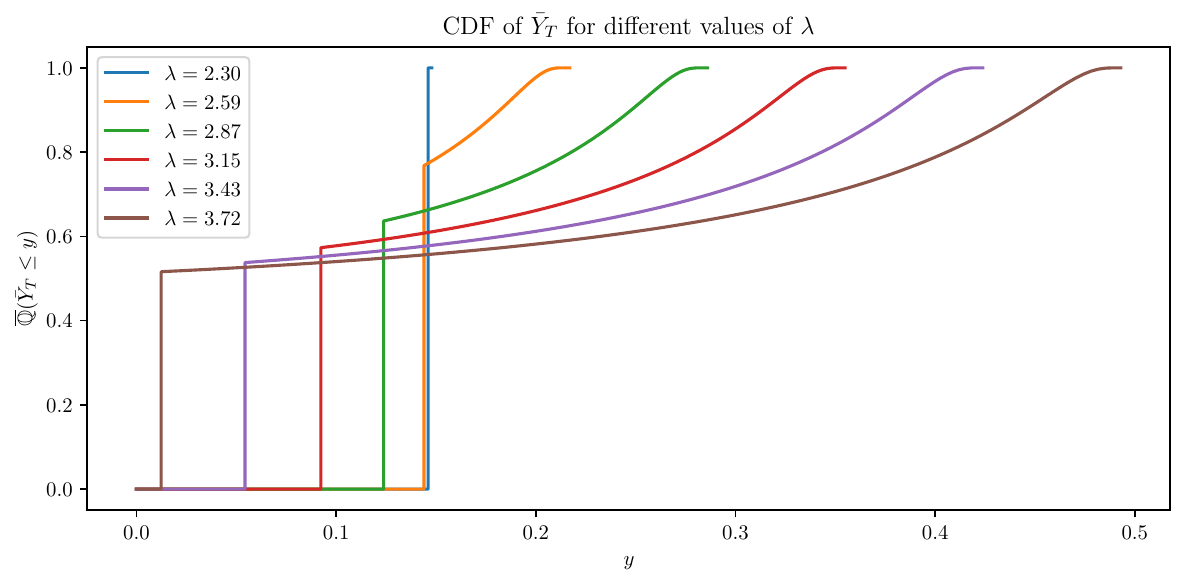}
  \caption{\label{fig:CDF_YT} The figure shows the cumulative distribution function of the optimal wealth $Y_T$ for the optimal strategy. In this case, different curves show the effect of different values of $\lambda$. The parameters are the same as in Figure~\ref{fig:MPlot}.}
\end{figure}

\section{Effect of static/dynamic hedging under intrinsic wealth constraints: One-touch options}\label{sec:one-touch}

In this section we consider the impact of intrinsic wealth constraints when hedging a path-dependent option. Specifically, in the case of a one-touch option, we are able to compare the case of hedging dynamically using the underlying asset alone, with the case of completing the hedge using a static position in vanilla options. In this case, we exploit the classical example of a static superhedge of a one-touch option using vanilla options and dynamic trading in the underlying asset due to Hobson~\cite{Hobson:1998aa}. As we will see, there is strong numerical evidence that even in a complete market setting, where the vanilla options are priced at their replication price, the impact of the static position on the intrinsic wealth constraint is significant, and leads to notable benefits to the trader who looks to sell a one-touch option.

\subsection{Max-plus representation for the one-touch option}

We begin by recalling the definition of the one-touch option, and Hobson's static superhedging strategy.

A one-touch option is a path-dependent option which pays out a fixed amount if the underlying asset price hits a pre-specified barrier level at any time before maturity. For simplicity, we suppose the barrier crossing is determined in forward price units, and our discount process is $D_t=e^{-rt}$, with $r$ a positive constant, under the hypotheses stated in the previous section. From Example~\ref{ex:Intrinsic}-\ref{item:3}, the \emph{one-touch} option has payoff $\OT^B_T := \ind_{\{S_T^* \ge B\}}$, where $S_t^* = \sup_{r \le t} S_r$ is the maximum process and $B >0$ is a fixed barrier. Recall that the discounted version of $S$ and $B$ is denoted by $S^D_t=D_tS_t$  and $B^D=D_T B$. Define $H_{B}=\inf\{t\geq 0: S_t^D\geq B^D\}$ and the one-touch option can alternatively be written as
\begin{align*}
    C^0_T&=\ind_{\{S_t\geq   BD_t^{-1}D_T,\;\text{for\;some}\;  t\in[0,T]\}}\\
    &=\ind_{\{S_t^D\geq   B^D,\;\text{for\;some}\;  t\in[0,T]\}}.
\end{align*}

Hobson's superhedging strategy for the one-touch option consists of a short position in the one-touch option itself, combined with a static position in vanilla Call options and dynamic trading in the underlying asset. Specifically, given  $K<B$ fixed, we can consider the portfolio with a payoff $\tilde{C}^0_T$ composed by a long position of  $\frac{1}{B-K}$ Call options with strike $K$  and, if $S^D$ hits $B^D$ before terminal time $T$, short sell $\frac{1}{B-K}$ units of asset, that is, 
\begin{align*}
\tilde{C}^0_T=\frac{1}{B-K}(S_T-K)_+ +\frac{1}{B-K}\{(BD_TD^{-1}_{H_B})D_{H_B}D_T^{_-1}-S_T\}\ind_{\{H_{B}\leq T\}}.
\end{align*}
Then, putting together the above expressions, the portfolio value at time $T$ is given by
\begin{align*}
\tilde{C}^0_T-C^0_T&=\frac{1}{B-K}[(S_T-K)_+ +(B-S_T)\ind_{H_B\leq T}]-\ind_{\{H_B\leq T\}}\\
&=
\begin{cases}
\frac{1}{B-K}(S_T-K)_+,\;\;&H_B>T\\
\frac{1}{B-K}(K-S_T)_+,\;\;&H_B\leq T.
\end{cases}
\end{align*}
Since the right-hand side is always non-negative, we deduce that $\tilde{C}^0_T$ is a superhedge for the one-touch option.
 
In \cite{Hobson:1998aa} it is shown moreover that the superhedging strategy is optimal in the sense that there exists a model under which $\tilde{C}^0_T = C^0_T$ almost surely, and hence the cost of the superhedge is equal to the arbitrage-free price of the one-touch option. Note that the choice of $K$ is not fixed, but there is a choice of $K$ which minimises the cost of the superhedge, and this choice of $K$ corresponds to the model which attains equality.

 Notice that from this expression we can deduce that its intrinsic value at time $t$ has the form
 \begin{align*}
 \In_t(\tilde{C}^0_T-C^0_T) &=
 \begin{cases}
% \frac{1}{B-K}(S_T-K)_+,\;\;&t=T<H_B\\
 \frac{D_T}{D_t}\cdot\frac{1}{B-K}(K-\frac{D_t}{D_T}S_t)_+,\;\;&H_B\leq t \leq T\\
 0, \;\;&t<H_B\wedge T 
\end{cases}\\
&=
\begin{cases}
% \frac{1}{B-K}(S_T-K)_+,\;\;&t=T<H_B\\
  \frac{1}{B-K}(K\frac{D_T}{D_t}-S_t)_+,\;\;&H_B\leq t\leq T\\
0, \;\;&t<H_B\wedge T ,
\end{cases}
\end{align*}
and hence
\begin{align*}
D_t \In_t(\tilde{C}^0_T-C^0_T) &=
\begin{cases}
% \frac{1}{B-K}(S_T-K)_+,\;\;&t=T<H_B\\
  \frac{1}{B-K}(K^D-S^D_t)_+,\;\;&H_B\leq t\leq T\\
0, \;\;&t<H_B\wedge T,
\end{cases}
 \end{align*}
 which is clearly a $\Q$-submartingale.
 Then,  the intrinsic wealth constraint process (see (\ref{eq:ZetaDef})),   is given by
  \begin{equation} \label{eq:ZetaDef2}
  \zeta_t := -\alpha D_t^{-1}\xi_t^{-1} - \xi_t^{-1} \In_t(\tilde{C}^0_T-C^0_T) + \Ebq{ \xi_T^{-1} (\tilde{C}^0_T-C^0_T) | \Fc_t},\;\;\;t\in[0,T),
  \end{equation}
with $\xi_t$ as in (\ref{eq:4}). Following the same arguments given in Remark \ref{rem:restprop}, we get that $ \zeta_t$ is a $\bQ$-supermartingale, and hence
$$
\hat{\zeta}_t:=D_t\xi_t\zeta_t=-\alpha-D_t \In_t(\tilde{C}^0_T-C^0_T) + \Eq{ D_T (\tilde{C}^0_T-C^0_T) | \Fc_t},\;\;t\in[0,T],
$$
is a $\Q$-supermartingale  with terminal condition $\hat{\zeta}_T=-\alpha$. To include the non-negativity constraint on the intrinsic wealth at the terminal time, as above we define 
 \begin{align*}
    \hat{\zeta}_t^0 &:=\hat{\zeta}_t  \ind_{\{t<T\}}=
    \begin{cases}
      \hat{\zeta}_t ,& t < T\\
      0, & t = T.
    \end{cases}
\end{align*}
 From the previous calculations of the intrinsic value of the derivative, for $t<T$, we have that
   \begin{align*}
   \hat{\zeta}_t^0  =&-\alpha- \frac{1}{B-K}(K^D-S^D_t)_+\ind_{\{H_B\leq t< T\}}\\
  &+ \frac{D_T}{B-K}\left\{\Eq{(S_T-K)_+ \ind_{\{H_B> T\}}+(K-S_T)_+\ind_{\{H_B\leq T\}}| \Fc_t}\right\}
  %& - \frac{1}{B-K}(K^D-S^D_t)_+\ind_{\{H_B\leq t< T\}}.
   \end{align*}

As a first step, we provide a Max-plus representation for the Snell envelope of  $ \hat{\zeta}_t^0$, denoted as $ \hat{\zeta}_t^{0,*}$, under measure $\Q$, meaning that there exists a process $J_u^\zeta$ such that
\[
 \hat{\zeta}_t^{0,*}=\Eq{\sup_{t \le u \le T} J_u^\zeta | \Fc_t}.
 \]
  Note that the Snell envelope of $ \hat{\zeta}^0$ is equal to the Snell envelope of $ \hat{\zeta}\vee 0$,  and using the last part of Lemma \ref{lemma:A1}, we can get the representation for the later once we have the one for  $\hat{\zeta}$. The ideas to do this have been already implemented in the proof of Theorem \ref{thm:long-position-call}, and 
  Corollary \ref{lem:rho-dec} for the case when $\alpha=0$,  therefore, we will simply outline the line of argument that should be followed.

First, using the fact that $\hat{\zeta}$ is a supermartingale, its Max-plus representation can be obtained using the first part of  Lemma   \ref{lemma:A1}, using the analogous  class of stopping times as in (\ref{def:stoppingtime}). Defining 
$$
\varphi_{\Q}(u)=\frac{1}{B-K}
\Eq{(K^D-S^D_T)_+\;|\;S_u^D=K^D}-\alpha,
$$
it is given by
$$
  J_u^{\hat{\zeta}}=
   \begin{cases}
   -\alpha, &u=T\\
   %\left(-\alpha+\frac{D_T}{B-K}(S_T-K)_+ \right), &u\in [0,T), \;; H_B> T\\
   \varphi_{\Q}(u), &u=H^*<T,\\
   -\infty, &\text{otherwise}.
   \end{cases}
$$
% From here, we can now use   Lemma   \ref{lemma:A1} part (ii), to get the representation for  $ \hat{\zeta}_t^{0,*}$,
% $$
%   J_u^{\zeta^{0,*}}=
%    \begin{cases}
%    \left(-\alpha+\frac{D_T}{B-K}(S_T-K)_+ \right)_+ \ind_{\{H_B > T\}}, &u=T\\
%    \varphi_{\Q}(u), &u=H^*<T,\\
%    -\infty, &\text{otherwise},
%    \end{cases}
% $$
% where $H^*=\inf\{t\geq H_B\;|\;S_t^D=K^D\}$.  

From here, we can now use   Lemma   \ref{lemma:A1} part (ii), to get the representation for  $ \hat{\zeta}_t^{0,*}$,
$$
  J_u^{\zeta^{0,*}}=
   \begin{cases}
   \left(-\alpha+\frac{D_T}{B-K}(S_T-K)_+ \right)_+ \ind_{\{H_B > T\}}, &u=T\\
   \varphi_{\Q}^*(u), &H_B < u<T, S_u^D=K^D\\
   -\infty, &\text{otherwise},
   \end{cases}
$$
where 
$$
\varphi_{\Q}^*(u) = \max\left\{\varphi_{\Q}(u),0\right\}.
$$

In this case we can apply Lemma~\ref{lemma:A1} part (i) with $\tau_t := \inf\{u \ge t : S_u^D = K^D, H_B \le u < T\} \wedge T$.

 \begin{remark}
Notice that  the Max-plus representation is needed  under measure $\bQ$, and it can be obtained using Lemma \ref{lem:MaxPlusLemma} part (ii), applying to the corresponding Radon-Nikodym derivative martingale $\{D_t\xi_t\}$; see Remark \ref{rem:restprop}. This result was not needed in the proof of Theorem \ref{thm:long-position-call}
because the process $\{\phi_t := \xi_t^{-1} D_t^{-1}\}$ was implicit in all the calculations where it was involved; see, in particular,  \ref{def:z}, so that we could obtain the representation under  the measure  $\bQ$.
 
\end{remark}

%This follows from  Lemma \ref{lem:MaxPlusLemma} part (i), taking ...
 
%A similar calculation can be done in the case where the trader does not hold the static position in Call options. In this case, we get \dots

\subsection{Numerical results}

We now present some numerical results which illustrate the impact of the intrinsic wealth constraint when hedging a one-touch option, comparing the case where the trader uses only dynamic trading in the underlying asset, with the case where the trader also holds a static position in vanilla Call options as described above.

Using the max-plus representation described above, we can numerically compute the optimal terminal portfolio wealth for the trader in both cases. Computing this requires numerical evaluation of hitting time densities for Brownian motion with drift, which can be done using standard numerical methods for inversion of the Laplace transform. 

In the numerical results below, we consider a trader who sells one-touch options with barrier $B=1.9$ and maturity $T=2$, in a Black-Scholes-Merton model with parameters $S_0=1.2, r=0.01, \sigma=0.5, \theta = 0.05$. The trader has initial wealth $w_0=0.1$, risk aversion parameter $p=0.75$, and intrinsic wealth constraint parameter $\alpha=0.1$. In this model, the arbitrage-free price of the one-touch option is approximately $0.41$.

Since we believe the trader will benefit from holding a short position in the one-touch option, we consider the case where the trader is able to sell the one-touch option for a premium to the cost of dynamically replicating. Specifically, we suppose that the trader is able to sell the one-touch option for a premium of $\Delta C = 0.02$ above the replication cost.

\begin{figure}
    \centering
    \includegraphics[width=\textwidth]{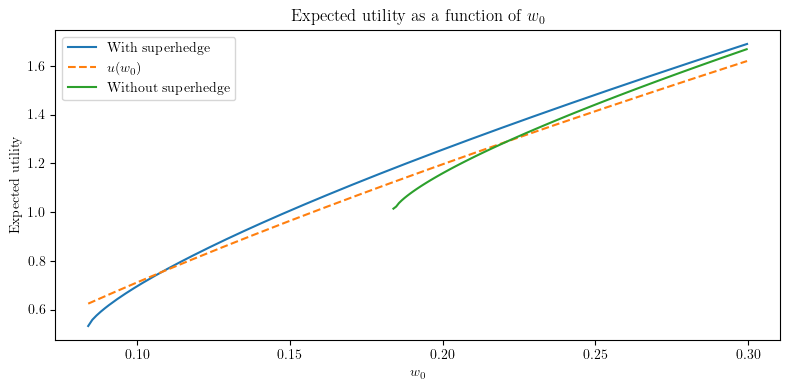}
    \caption{\label{fig:OT_UtilityPlot} The figure shows the expected utility of the trader as a function of the trader's initial wealth, $w_0$. The solid blue line shows the expected utility when the trader holds a static position in vanilla Call options as part of a superhedging strategy for the one-touch option, while the solid green line shows the expected utility when the trader does not hold the static position. The dashed orange line shows the utility of the trader if they were not to sell the one-touch option at all. Here the superhedge uses call options with strike $K=0.1.3$. }
\end{figure}

In Figure~\ref{fig:OT_UtilityPlot}, we can see the expected utility of the trader under different scenarios. We observe that when the trader's initial wealth is small, the benefit of selling the call option does not outweigh the loss that is incurred because the trader must trade in such a way as to avoid breaching the intrinsic wealth constraint. When the wealth is large, both the case with and without the semi-static hedge outperform the case where the trader does not sell the one-touch option, and in both cases the trader is able to exploit the premium received from selling the one-touch at a premium with minimal effect of the intrinsic wealth constraint. Notably, however, the case where the trader holds the semi-static position in vanilla Call options is much more impactful in the case where the trader's initial wealth is moderate. Moreoever, the initial wealth at which the trader is even able to implement the strategy is much lower in the semi-static case ($w_0 \approx 0.08$) than in the case without the static hedge ($w_0 \approx 0.18$). 

\begin{figure}
    \centering
    \includegraphics[width=\textwidth]{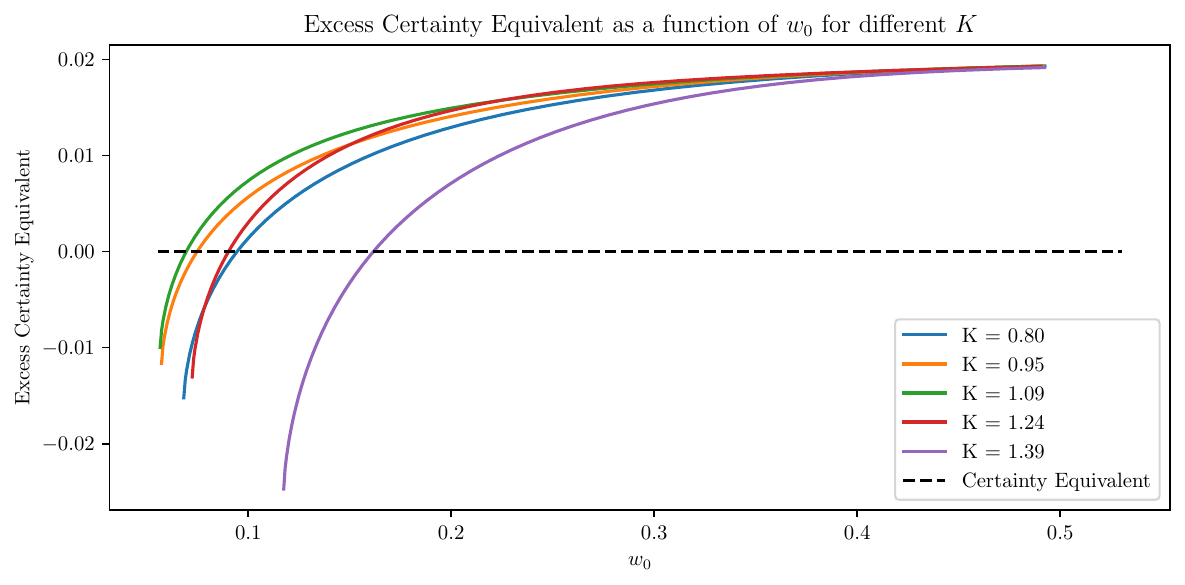}
    \caption{\label{fig:OT_UtilityPlot_K} The figure shows the expected utility of the trader as a function of the wealth $w_0$ of trader for different strikes. The utiliy is shown in terms of the difference between the certainty equivalent of the utility with and without the hedge. The plots are shown for a range of stikes $K$ for the vanilla Call options used in the semi-static hedge. }
\end{figure}

\begin{figure}
    \centering
    \includegraphics[width=\textwidth]{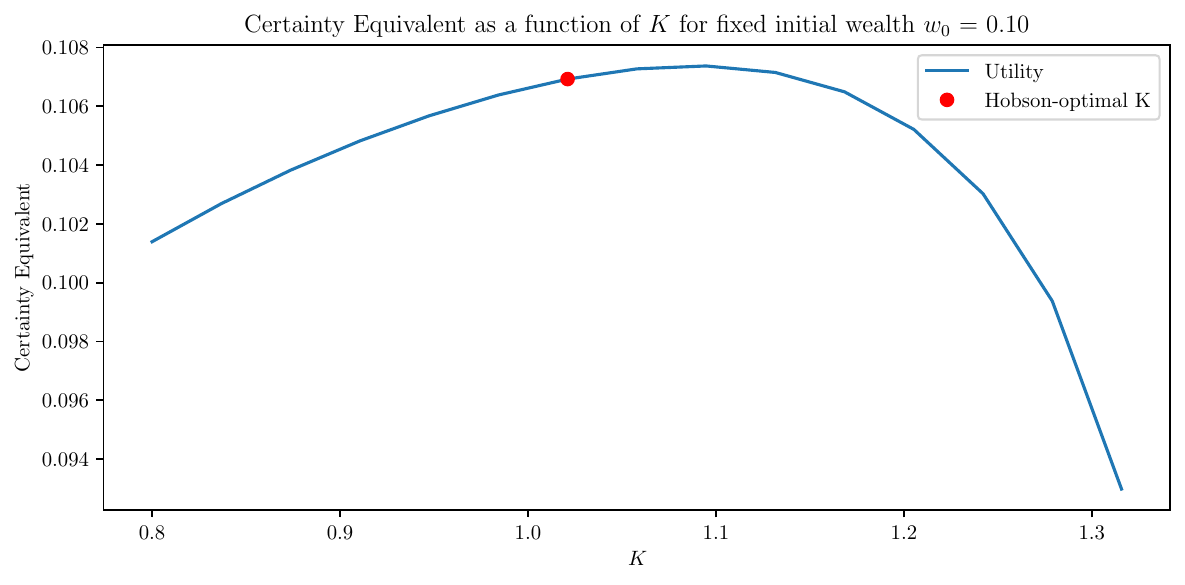}
    \caption{\label{fig:OT_CertaintyEquivalent_vs_K} The plot shows the certainty equivalent of the trader when holding the semi-static hedge as a function of the strike $K$ of the vanilla Call options used in the hedge. Highlighted is the choice of semi-static hedge which corresponds to the minimal-cost superhedge ('Hobson-optimal $K$'). The initial wealth is $w_0=0.1$. }
\end{figure}

We can also examine the impact of different choices of the parameter $K$. Figure~\ref{fig:OT_UtilityPlot_K} shows the expected utility of the trader as a function of initial wealth for a range of different strikes $K$ used in the semi-static hedge. We observe that there is no uniformly best choice of $K$. However there is some evidence in Figure~\ref{fig:OT_CertaintyEquivalent_vs_K} that suggests the optimal choice of $K$ is not exactly the `Hobson-optimal' strike, but this choice is close to optimal.

\section{Summary and Future Work}

In this paper we consider the problem of optimal investment in a portfolio which combines dynamic trading and a static position in options. The novelty in our work comes from a trading constraint which is based on the intrinsic, or worst case value of the option. This setting allows us to develop a framework for hedging which sits between classical and robust settings for option pricing. We are able to develop explicit characterisations of the optimal trading strategy in certain cases, and we are able to see in simple examples how the optimal strategy finds a balance between the desire to maximise expected profit, and the risk associated with extreme positions relative to the trader's capacity to sustain short-term mark-to-market losses.

Notably, we see in our numerical results that even in a complete market setting, the presence of the intrinsic wealth constraint can have a significant impact on the optimal strategy. In particular, in the case of the one-touch option, we are able to see the benefit of holding semi-static hedging positions in vanilla options, even when these options are priced at their replication cost. This suggests that even in complete market settings, the presence of intrinsic wealth constraints could justify  the use of semi-static hedging strategies.

While our explicit results make fairly strong assumptions (for example, complete markets), our framework is flexible, and future work to understand the impact of considering a larger class of possible hedging models (e.g. moving to an incomplete market setting), or allowing for uncertain volatility models, for example, would be interesting to understand.

\bigskip

\filbreak

\noindent
{\bf Acknowledgements:}

\noindent
AC and DHH acknowledge the support of the Royal Society, through the Newton International Fellowship scheme NI160069.  The second author is very grateful for the hospitality of the University of Bath.

\appendix

\section{Results on Max-Plus Representations}

In this appendix we prove some results on Max-plus representations for c\'adl\`ag supermartingales under specific assumptions on the structure of the supermartingale.

We suppose specifically that the c\'adl\`ag supermartingale $(X_t)_{t \in [0,T]}$ can be written as a martingale up to a specific exit time. That is, we suppose there exists a maximal family of increasing stopping times $\{\tau_t; t\in[0,T]\}$ such that $\tau_t \in [t,T]$, $\tau_t \le \tau_s$ for $t \le s$, 
\begin{equation} \label{eq:TauCond}
X_t = \Eq{X_{\tau_t} | \Fc_t},
\end{equation}
and such that 
\begin{equation} \label{eq:TauCond2}
  X_t > \Eq{X_{\sigma} | \Fc_t}, \text{ for all stopping times } \sigma \ge \tau_t \text{ with } \P(\sigma > \tau_t) > 0. 
\end{equation} 
In canonical cases, we may consider $\tau_t$ to be the first hitting time after time $t$ by a process to a specific region, for example. We note that the case where $\tau_t = t$ is not excluded, although in the following result it will imply a very specific structure for $X$.

Then we have the following result.

\begin{lemma}\label{lemma:A1}
  Suppose there exists a family of increasing stopping times $\{\tau_t; t\in[0,T]\}$ such that $\tau_t \in [t,T]$ \eqref{eq:TauCond} and \eqref{eq:TauCond2} hold, and a decreasing function $\varphi:[0,T] \to \R$ such that if $\tau_t = s$ for some $t$, then
  \begin{equation*}
    X_s \ge \varphi(s) \ge X_u, \quad \text{ for all } u > s.
  \end{equation*}
  Then $(X_t)_{t \in [0,T]}$ has Max-plus representation
  \begin{equation*}
    X_t = \Eq{\sup_{t \le u \le T} J_u | \Fc_t},
  \end{equation*}
  where
  \begin{equation*}
    J_u = 
    \begin{cases}
      X_{\tau_u}, & u= \tau_u < T, \\
      X_T, & u = T, \\
      -\infty, & \text{ otherwise}.
    \end{cases}
  \end{equation*}
  Moreover, the smallest (in convex increasing order) martingale dominating $(X_t)_{t \in [0,T]}$ is given by
  \begin{equation*}
    M_t = \Eq{\sup_{0 \le u \le T} J_u | \Fc_t},
  \end{equation*}
  and for any constant $c \in \R$, the smallest (in convex increasing order) martingale dominating $(X_t \vee c)_{t \in [0,T]}$ is given by
  \begin{equation*}
    M_t^c = \Eq{\left(\sup_{0 \le u \le T} J_u\right) \vee c | \Fc_t},
  \end{equation*}
  that is, the Max-plus representation is given by $J_u^c = J_u \vee c$, or equivalently
  \begin{equation*}
    J_u = 
    \begin{cases}
      X_{\tau_u}, & u= \tau_u < T, X_{\tau_u} \ge c \\
      X_T \vee c, & u = T, \\
      -\infty, & \text{ otherwise}.
    \end{cases}
  \end{equation*}
\end{lemma}

Observe that in the trivial case where $\tau_t = t$, then the conditions of the lemma imply that $X$ is a decreasing process, and hence the Max-plus representation is trivially given by $J_u = X_u$ for $u \le T$.

\begin{proof}
  First observe that defining 
  \begin{equation*}
    \tilde{J}_u = 
    \begin{cases}
      X_{u}, & u= \tau_s < T, \text{ some } s \le u, \\
      X_T, & u = T, \\
      -\infty, & \text{ otherwise}.
    \end{cases}
  \end{equation*}
  then
  \begin{equation*}
    X_t = \Eq{\sup_{t \le u \le T} \tilde{J}_u | \Fc_t},
  \end{equation*}
  since if $s = \tau_t$ for some $t \le s$ then for $u \ge s$ we have $X_u \le \varphi(s) \le X_s$, and hence $X_s = \Eq{X_{\tau_s}| \Fc_s} \ge \Eq{X_u | \Fc_s}$ with equality if and only if $X_{\tau_s} = X_u = X_s$ almost surely. In particular, $\sup_{t \le u \le T} \tilde{J}_u = X_s$. On the other hand, if $s \neq \tau_t$ for some $t \le s$, then in particular $s < \tau_s$, and $X_s = \Eq{X_{\tau_s} | \Fc_s} = \Eq{\tilde{J}_{\tau_s} | \Fc_s} \ge \Eq{\sup_{\tau_s \le u \le T} \tilde{J}_u | \Fc_s}$. Moreover, by the maximality of $\tau_t$, if there exists $t < s$ such that $s < \tau_t$ with positive probability, then $\tau_t = \tau_s$ on $\{s < \tau_t\}$ since $X_s = \Eq{X_{\tau_s \wedge \tau_t}|\Fc_s} = \Eq{X_{\tau_t}|\Fc_s}$ on this set.

  The rest of the proof now follows immediately from the fact that $X_t = J_t = J^c_t$ whenever $J_t > c$, and properties of max-plus martingales.

  % The Max-plus representation follows directly from the assumptions on the structure of $X$ and the definition of $\tau_t$. In particular, if $\tau_t = s$, then for $u \ge s$ we have $X_u \le X_s$, and hence $X_s = \Ep{X_{\tau_s}| \Fc_s} \le \Ep{X_s | \Fc_s}$ with equality if and only if $X_{\tau_s} = X_s$ almost surely.
    
  % For the final claim, observe that since $M_t^c$ is the smallest martingale dominating $X_t \vee c$, then $M_t^c \le M_t$ for all $t$. Hence, we have
  % \begin{align*}
  %   M_t^c & = \Eq{\sup_{0 \le u \le T} J_u^c | \Fc_t} \le \Eq{\left(\sup_{0 \le u \le T} J_u\right) \vee c | \Fc_t} \le M_t^c,
  % \end{align*}
  % and hence equality holds throughout.
\end{proof}

To help us in the verification of some results, we will make use of the following result.
\begin{lemma}\label{lem:MaxPlusLemma}
  Let $(X_t)_{t\in[0,T]}$ be a $\Q$-supermartingale with c\`adl\`ag paths. %and Max-plus respresentation $J$.
  \begin{enumerate}
    \item Suppose $X = Y + Z$ where $Y$ and $Z$ are both also c\`adl\`ag $\Q$-supermartingales such that $Y_T, Z_T \ge 0$. Suppose in addition that there exists a stopping time $\tau \le T$ and an $\Fc_{\tau}$-measurable set $A$ such that:
    \begin{enumerate}
      \item $X_{t \wedge \tau}$ is a $\Q$-martingale;
      \item $Y_t = 0$ on $A$ for all $t \ge \tau$;
      \item $Z_t = 0$ on $A^c$ for all $t \ge \tau$. 
    \end{enumerate}
    Then $X_t = \Eq{\sup_{t \le u \le T} J_u | \Fc_t}$, where
    \begin{align*}
      J_u & = J_u^Y + J_u^Z,
    \end{align*}
    and $J_u^Y$ and $J_u^Z$ are the Max-plus representations of $Y$ and $Z$ respectively.

    \item Let $(X_t)_{t\in[0,T]}$ be a $\Q$-supermartingale with Max-plus representation $J_u^X$, and suppose $\overline{\Q}$ is an equivalent probability measure to $\Q$ with Radon-Nikodym derivative $\frac{\di \overline{\Q}}{\di \Q} = M_T$, where $(M_t)_{t\in[0,T]}$ is a strictly positive $\Q$-martingale.
    
    Suppose in addition that there exists a maximal family of stopping times $\{\tau_t; t\in[0,T]\}$ such that $\tau_t \in [t,T]$ for each $t$, $\tau_t \le \tau_s$ when $t \le s$ and $X_t = \Eq{J_{\tau_t}^X | \Fc_t}$. If in addition
    \begin{align*}
      J_{\tau_t}^X M^{-1}_{\tau_t} & \ge J_{u}^X M^{-1}_{u}, \quad u \in [t,T], 
    \end{align*}
    then $(XM^{-1})_{t\in[0,T]}$ is a $\overline{\Q}$-supermartingale with Max-plus representation $J_u^{\overline{X}} = \frac{1}{M_u} J_u^X$, i.e.
    \begin{align*}
      X_t & =  M_t \Ebq{\sup_{t \le u \le T} \left(J_u^{X} M_u^{-1}\right) | \Fc_t}.
    \end{align*}
  \end{enumerate}
\end{lemma}

\begin{proof}
\begin{enumerate}
  \item First observe since $Y$ and $Z$ are non-negative, we have $J_T^Y, J_T^Z \ge 0$, and hence $J_T = J_T^Y + J_T^Z \ge 0$. Further, it follows from the fact that $X_{t \wedge \tau}$ is a $\Q$-martingale that both $Y_{t \wedge \tau}$ and $Z_{t \wedge \tau}$ are also $\Q$-martingales. Therefore, for $t \le \tau$, we have almost surely
  \begin{align*}
    J_u^Y & \le \sup_{\tau \le v \le T} J_v^Y, \quad u \in [0,\tau],
  \end{align*} 
  and similarly for $J_u^Z$. In addition, on $A$, for $u \ge \tau$ we have $J_u^Z = 0$, and similarly on $A^c$ for $J_u^Y$. Hence
  \begin{align*}
    \sup_{t \le u \le T} (J_u^Y + J_u^Z) & = \sup_{t \vee \tau \le u \le T} (J_u^Y + J_u^Z) \\
    & =
    \begin{cases}
      \sup_{t \vee \tau \le u \le T} J_u^Y, & \text{ on } A, \\
      \sup_{t \vee \tau \le u \le T} J_u^Z, & \text{ on } A^c,
    \end{cases}
    \\
    & =
      \sup_{\tau \vee t \le u \le T} J_u^Y + \sup_{\tau \vee t \le u \le T} J_u^Z,\\
    & = \sup_{t \le u \le T} J_u^Y + \sup_{t \le u \le T} J_u^Z.
  \end{align*}  

  \item Since $X_t = \Eq{J_{\tau_t}^X | \Fc_t}$ we have 
  \begin{align*}
    M_t^{-1} X_t & = M_t^{-1} \Eq{J_{\tau_t}^X | \Fc_t} \\
    & = \Ebq{J_{\tau_t}^X M_{T}^{-1} | \Fc_t} \\
    & = \Ebq{J_{\tau_t}^X M_{\tau_t}^{-1} | \Fc_t}\\
    & = \Ebq{\sup_{t \le u \le T} (J_u^X M_u^{-1}) | \Fc_t},
  \end{align*}
  where in the second last step we have used the fact that $M^{-1}$ is
  a $\bQ$-martingale, and in the final step we have used the
  assumption that for all $u \in [t,T]$,
  $J_{\tau_t}^X M^{-1}_{\tau_t} \ge J_{u}^X M^{-1}_{u}$.
\end{enumerate}

\end{proof}

%\comment{TODO:

%1) Compute the form of the terminal wealth for the Black-Scholes-Merton model ... in general? In specific case? ($\xi_t \equiv 1$)/assumption on `decreasing' value on boundary'.

%2) As above, but for Bachelier model (specific formulae...)

%3) Variational Analysis for these cases. Finding the optimal number of units held.

%4) One-touch options.

%5) Future Work: G-Expectations \& different ranges for agent, intrinsic value?
%}

%% BIBLATEX STUFF
% \renewbibmacro*{in:}{}

% \printbibliography
%% END OF BIBLATEX STUFF

\bibliographystyle{plain}
\bibliography{UtilityPlus}

\end{document}